\documentclass[10pt]{article}

\usepackage[utf8]{inputenc}

\usepackage{epsf}
\usepackage{amsmath}

\allowdisplaybreaks

\usepackage[showframe=false]{geometry}
\usepackage{changepage}

\usepackage{epsfig}
\usepackage{amssymb}

\usepackage{amsthm}
\usepackage{setspace}
\usepackage{cite}
\usepackage{mcite}

\usepackage{algorithmic}  
\usepackage{algorithm}

\usepackage{shadow}
\usepackage{fancybox}
\usepackage{fancyhdr}

\usepackage{color}
\usepackage[usenames,dvipsnames,svgnames,table]{xcolor}

\definecolor{mypurple}{rgb}{.4,.0,.5}

\usepackage[hyphens]{url}

\usepackage[colorlinks=true,
            linkcolor=black,
            urlcolor=blue,
            citecolor=purple]{hyperref}

\usepackage{breakurl}

\def\y{{\bf y}}

\def\x{{\bf x}}

\def\x{{\mathbf x}}

\def\u{{\bf u}}

\def\x{{\bf x}}
\def\y{{\bf y}}
\def\z{{\bf z}}
\def\q{{\bf q}}
\def\m{{\bf m}}

\def\c{{\bf c}}

\def\h{{\bf h}}

\def\tr{\mbox{Tr}}

\def\tr{{\rm tr}\,}

\def\cH{{\mathcal H}}

\def\be{\begin{equation}}
\def\ee{\end{equation}}
\def\ba{\left[\begin{array}}
\def\ea{\end{array}\right]}

\def\u{{\bf u}}

\def\x{{\bf x}}
\def\y{{\bf y}}
\def\z{{\bf z}}
\def\q{{\bf q}}

\def\c{{\bf c}}

\def\p{{\bf p}}

\def\1{{\bf 1}}

\def\0{{\bf 0}}

\def\calX{{\cal X}}
\def\calY{{\cal Y}}

\def\mR{{\mathbb R}}
\def\mN{{\mathbb N}}
\def\mE{{\mathbb E}}
\def\mS{{\mathbb S}}

\def\lp{\left (}
\def\rp{\right )}

\def\y{{\bf y}}

\def\x{{\bf x}}

\def\x{{\mathbf x}}

\def\u{{\bf u}}

\def\x{{\bf x}}
\def\y{{\bf y}}
\def\z{{\bf z}}
\def\q{{\bf q}}

\def\c{{\bf c}}

\def\h{{\bf h}}

\def\tr{\mbox{Tr}}

\def\tr{{\rm tr}\,}

\def\cH{{\cal H}}

\def\be{\begin{equation}}
\def\ee{\end{equation}}
\def\ba{\left[\begin{array}}
\def\ea{\end{array}\right]}

\def\u{{\bf u}}

\def\x{{\bf x}}
\def\y{{\bf y}}
\def\z{{\bf z}}
\def\q{{\bf q}}

\def\c{{\bf c}}

\def\p{{\bf p}}

\def\({\left (}
\def\){\right )}

\def\1{{\bf 1}}
\def\m{{\bf m}}
\def\q{{\bf q}}

\def\0{{\bf 0}}

\def\cX{{\mathcal X}}
\def\cY{{\mathcal Y}}

\def\rx{\x^{(r)}}
\def\rcX{{\mathcal X}^{(r)}}

\def\bq{\bar{q}}

\newtheorem{theorem}{Theorem}

\newtheorem{corollary}{Corollary}

\newtheorem{observation}{Observation}

\begin{document}

\begin{singlespace}

\title {Generic proving of replica symmetry breaking 
}
\author{
\textsc{Mihailo Stojnic
\footnote{e-mail: {\tt flatoyer@gmail.com}} }}
\date{}
\maketitle

\centerline{{\bf Abstract}} \vspace*{0.1in}

We study the replica symmetry breaking (rsb) concepts on a generic level through the prism of recently introduced interpolating/comparison mechanisms for bilinearly indexed (bli) random processes. In particular, \cite{Stojnicnflgscompyx23} introduced a \emph{fully lifted} (fl) interpolating mechanism  and \cite{Stojnicsflgscompyx23} developed its a \emph{stationarized fully lifted} (sfl) variant. Here, we present a sfl \emph{matching} mechanism that shows that the results obtained in \cite{Stojnicnflgscompyx23,Stojnicsflgscompyx23} completely correspond to the ones obtained by a statistical physics replica tool with the replica symmetry breaking (rsb) form suggested by Parisi in \cite{Par79,Parisi80,Par80}. The results are very generic as they allow to handle pretty much all bilinear models at once. Moreover, given that the results of \cite{Stojnicnflgscompyx23,Stojnicsflgscompyx23} are extendable to many other forms, the concepts presented here automatically extend to any such forms as well.

\vspace*{0.25in} \noindent {\bf Index Terms: Random processes; Replica symmetry breaking; Interpolation; Stationarization}.

\end{singlespace}

\section{Introduction}
\label{sec:back}

Studying random mathematical structures has been the subject of extensive interest over the last several decades in a host of different scientific fields. These range from highly theoretical to very practical ones and include various disciplines of theoretical and applied mathematics and probability (see, e.g., \cite{Talbook11a,Talbook11b,Pan10,Pan10a,Pan13,Pan13a,Pan14,Guerra03,Tal06,StojnicRegRndDlt10,Stojnicgscompyx16,Stojnicgscomp16,SchTir03}), statistical physics (see, e.g., \cite{Par79,Par80,Par83,Parisi80,SheKir72,Little74,Gar88,GarDer88,KraMez89}), signal processing, information theory and algorithms (see, e.g., \cite{BayMon10,BayMon10lasso,StojnicCSetam09,StojnicGenLasso10}), neural and biological networks (see, e.g., \cite{AgiAstBarBurUgu12,AgiBarBarGalGueMoa12,AgiBarGalGueMoa12,Hebb49,Gar88,GarDer88,KraMez89}, theoretical computer science
(see, e.g., \cite{Moll12,CojOgh14,DinSlySun15,MezParZec02,MerMezZec06,FriWor05,AchPer04,AlmSor02,Wast12,LinWa04,Talbook11a,Talbook11b,NaiPraSha05,CopSor02,Ald01}) and many others. Many reasons are typically behind studying such structures. On a most general level, they can be split into two main categories: 1) randomness naturally arises in the phenomena of interest; and 2) randomness is usually imposed either as an attempt to resemble/emulate the true behavior of the studied objects or as a helping/facilitating tool that enables mathematical studying and analytical characterization of otherwise non-random objects. For problems motivated by the reasons from any of the two categories, many excellent results have been obtained over the last several decades while critically relying on a fruitful connection between approaches based on statistical physics and corresponding fully rigorous mathematics. The most prominent such connections include full settling of a host of well known problems related to the so-called Sherrington-Kirkpatrick (SK) models (see, e.g., \cite{Talbook11a,Talbook11b,Pan10,Pan10a,Pan13,Pan13a,Pan14,Guerra03,Tal06}), spherical perceptrons (see, e.g., \cite{SchTir03,StojnicGardGen13,Tal05,Talbook11a,Talbook11b,Wendel,Winder,Winder61,Cover65,DonTan09Univ}), compressed sensing/statistical regression (see, e.g., \cite{StojnicUpper10,StojnicCSetam09,StojnicGenLasso10,BayMon10,BayMon10lasso,DonohoPol}), satisfiability/assignment problems
(see, e.g., \cite{CojOgh14,DinSlySun14,DinSlySun15,DinSlySun14,MezParZec02,MerMezZec06,FriWor05,AchPer04}),
assignment/matching problems
(see, e.g., \cite{AlmSor02,Wast04,Wast09,Wast12,LinWa04,TalAssignment03,Talbook11a,Talbook11b,NaiPraSha05,CopSor02,CopSor99,Ald01,Ald92,MezPar87})  and so on.

A particular statistical physics tool, called \emph{replica theory}, is an integral part of studying of almost all of the above problems. Within the replica theory, one usually distinguishes two types of prevalent behavior: 1) the standard \emph{replica symmetric} one; and 2) its less conventional \emph{replica-symmetry breaking} (rsb) alternative. Also, among the rsb ones, two groups are most prevalent: i) the 1-rsb; and ii) full $r$-rsb with $r\rightarrow\infty$. While finite $r$, $r$-rsb scenarios are possible, they fairly rarely happen in the known problems. With the exception of the SK results \cite{Talbook11a,Talbook11b,Pan10,Pan10a,Pan13,Pan13a,Pan14,Guerra03,Tal06}), pretty much all of the above mentioned known mathematically rigorous results are either related to the replica symmetric scenarios ($0$-rsb) or to the $1$-rsb. This fact on its own already hints at two possible reasons for such rare occurrence of the full $r$-rsb scenarios: 1) the underlying problems are mostly of such nature that $0$- or $1$-rsb are sufficient to handle them; and/or 2) the underlying problems do require full $r$-rsb but it is difficult to rigorously prove such a necessity. Both of these reasons are likely correct. It is not inconceivable that a majority of the interesting problems became interesting due to some form of simplicity which might very well manifest itself in such internal structures that are indeed simple enough to be handled by $0$- or $1$-rsb. On the other hand, those that so far presented themselves as less interesting or hard are possibly of such internal structure that require higher orders or, even more likely, the full $r$-rsb.

As the need for making advancement in studying rsb is rather apparent, we here focus on those scenarios. Various analytical tools have been developed to mathematically rigorously attack these problems over the last several decades. Of our particular interest in this paper is the connection between the rsb and the comparisons of random processes. As is well known, studying random processes played a very strong (quite likely the key) role in establishing fully rigorous analytical characterizations of the SK models (see, e.g., \cite{Talbook11a,Talbook11b,Pan10,Pan10a,Pan13,Pan13a,Pan14,Guerra03,Tal06}). Their role in studying \emph{phase-transition} (PT) phenomena in compressed sensing (see, e.g., \cite{StojnicISIT2010binary,StojnicCSetam09,StojnicUpper10,StojnicICASSP10block,StojnicICASSP10var,StojnicICASSP10knownsupp,StojnicICASSP10knownsupp,StojnicLiftStrSec13,StojnicRicBnds13}), neural networks (see, e.g., \cite{StojnicGardSphNeg13,StojnicGardSphErr13,StojnicDiscPercp13}), and statistical mechanics problems (see, e.g., \cite{StojnicMoreSophHopBnds10,StojnicAsymmLittBnds11}) is rather apparent as well.


Studying random processes is, of course, topic on its own  (more on the development, importance, and relevant prior and contemporary considerations can be found in, e.g.,  \cite{Sudakov71,Fernique74,Fernique75,Kahane86,Stojnicgscomp16,Adler90,Lifshits85,LedTal91,Tal05,Gordon85,Slep62}). Many of their great properties have been uncovered ensuring that circumventing them in the analysis of random structures and optimization problems remains almost unimaginable. Despite a significant progress over the last a couple of decades, characterizations of many random structures and various PT phenomena remain out of reach (e.g., \cite{StojnicLiftStrSec13,StojnicMoreSophHopBnds10,StojnicRicBnds13,StojnicAsymmLittBnds11,StojnicGardSphNeg13,StojnicGardSphErr13} and references therein discuss various forms, including both quadratic and bilinear max and minmax ones, where satisfactory PT characterizations are still lacking). As noted in \cite{Stojnicgscompyx16}, any further progress on this front seems to be fully correlated to the progress in studying and understanding the underlying random processes. The mechanisms  of \cite{Stojnicgscomp16,Stojnicgscompyx16}  introduced the so-called \emph{partially lifting strategy} and made a strong step forward in that direction. More recently, \cite{Stojnicnflgscompyx23,Stojnicsflgscompyx23}  moved things to another level and introduced  a generic statistical interpolating/comparison mechanism called \emph{fully lifted} (fl) and its \emph{stationarization} along the interpolating path realization called \emph{stationarized fully lifted} (sfl). Here we make use of these concepts and study how they behave when compared to the rsb predictions. For concreteness, we choose bilinearly indexed random processes and corresponding statistical mechanics models based on the so-called bilinear Hamiltonians. We first present a collection of full rsb results that can be obtained for such models relying on the so-called \emph{Parisi rsb ansatz} introduced in a sequence of breakthrough papers \cite{Par79,Par80,Par83,Parisi80}. We then proceed by showing that the concepts of \cite{Stojnicnflgscompyx23,Stojnicsflgscompyx23} are strong enough to achieve results that \emph{exactly match} the full rsb predictions.

\section{General bilinear model}
\label{sec:sqrtposhop}

We consider a general bilinear model. In Section \ref{sec:mathdescsqrtposhop}, we start by first introducing a mathematical description of the model and then in Section \ref{sec:repsqrtposhop}, we present a replica approach that can be used for studying such a model.

\subsection{Mathematical description of bilinear Hamiltonian based models}
\label{sec:mathdescsqrtposhop}

As we are focused on a statistical physics approach to studying general bilinear models, the following Hamiltonian is of key interest
\begin{equation}
\cH(G)=\y^TG\x,\label{eq:hamsqrt}
\end{equation}
where $G\in\mR^{m\times n}$ are the so-called quenched interactions. We typically consider scenarios where $m$ and $n$ are large and $\alpha=\lim_{n\rightarrow \infty}\frac{m}{n}$ (with $\alpha$, basically, remaining a constant as $n\rightarrow \infty$). The behavior of various physical interpretations that can be described through the above Hamiltonian are usually analyzed via the so-called partition function
\begin{equation}
Z_{sq}(\beta,G) = \sum_{\x\in\cX}\sum_{\y\in \cY}e^{\beta\cH(G)}= \sum_{\x\in\cX}\sum_{\y\in \cY}e^{\beta \y^TG\x},\label{eq:partfunsqrt}
\end{equation}
where $\beta>0$ is a parameter, typically called the inverse temperature.  For the concreteness, we set $\calX=\{\x^{(1)},\x^{(2)},\dots,\x^{(l)}\}$ with $\x^{(i)}\in \mR^n$ and $\calY=\{\y^{(1)},\y^{(2)},\dots,\y^{(l)}\}$ with $\y^{(i)}\in \mR^m$. Also, for the convenience of the presentation, we skip adding a factor $\frac{1}{2}$ typically seen in front of the Hamiltonian in statistical physics literature. It goes without saying, that the bilinear Hamiltonian from (\ref{eq:hamsqrt}) and its associated partition function from (\ref{eq:partfunsqrt}) are very generic and encompass many well known particular examples often studied in statistical physics. While the complete list of all included specializations is pretty much endless, we here single out a few most famous ones that are directly obtainable relying on the bilinear Hamiltonian from (\ref{eq:hamsqrt}) : 1) For $m=n$ and $\cY=\cX=\{-\frac{1}{\sqrt{n}},\frac{1}{\sqrt{n}}\}$, one has precisely the so-called Sherrington-Kirkpatrick (SK) model (see, e.g., \cite{SheKir72,Pan13a,Parisi80,Par80,Guerra03,Tal06}); 2) For $\cX=\{-\frac{1}{\sqrt{n}},\frac{1}{\sqrt{n}}\}$ and $\cY\in\mS^m$, one has the so-called square-root Hopfield model  (see, e.g., \cite{Hop82,PasFig78,Hebb49,PasShchTir94,ShchTir93,BarGenGueTan10,BarGenGueTan12,Tal98,StojnicMoreSophHopBnds10}); 3) For $\cY=\{-\frac{1}{\sqrt{n}},\frac{1}{\sqrt{n}}\}$ and $\cX=\{-\frac{1}{\sqrt{n}},\frac{1}{\sqrt{n}}\}$ one obtains the asymmetric Little model (see, e.g., \cite{Little74,StojnicAsymmLittBnds11,Stojnicnflgscompyx23,Stojnicsflgscompyx23}); 4) For $\cY=\mS_+^m$ (where $\mS_+^m$ is the positive orthant part of the $m$-dimensional unit sphere) one has the so-called positive perceptrons; In particular for $\cX=\{-\frac{1}{\sqrt{n}},\frac{1}{\sqrt{n}}\}$ one obtains the positive \emph{binary} perceptron (see, e.g., \cite{StojnicGardGen13,GarDer88,Gar88,StojnicDiscPercp13,KraMez89,GutSte90,KimRoc98}) whereas for  $\cX=\mS^n$ one obtains the positive \emph{spherical} perceptron (see, e.g., \cite{StojnicGardGen13,StojnicGardSphErr13,StojnicGardSphNeg13,GarDer88,Gar88,Schlafli,Cover65,Winder,Winder61,Wendel62,Cameron60,Joseph60,BalVen87,Ven86,SchTir02,SchTir03}).

Instead of dealing directly with the partition function, $Z(\beta,G)$, one usually considers the following, more appropriate, scaled $\log$, (average) version called (average) free energy
\begin{equation}
f_{sq}(\beta)=\lim_{n\rightarrow\infty}\frac{\mE_G\log{(Z_{sq}(\beta,G)})}{\beta \sqrt{n}}=\frac{\mE_G\log{(\sum_{\x\in\cX}\sum_{\y\in \cY} e^{\beta\cH(G)})}}{\beta \sqrt{n}},\label{eq:logpartfunsqrt}
\end{equation}
where $\mE_G$ stands for the expectation with respect to $G$ (throughout the paper, $\mE$ always stands for the expectation, and the subscript indicates the type of randomness with respect to which the expectation is taken). The thermodynamic limit ground state regime is typically of particular interest. In addition to $n\rightarrow\infty$, in such a regime, one also has
$\beta\rightarrow\infty$ and
\begin{eqnarray}
\lim_{\beta\rightarrow\infty}f_{sq}(\beta) & = &
\lim_{\beta,\rightarrow\infty}\frac{\mE_G\log{(Z_{sq}(\beta,G)})}{\beta \sqrt{n}}.\label{eq:limlogpartfunsqrt}
\end{eqnarray}
The main mathematical object that we study in this section is precisely the free energy given in (\ref{eq:logpartfunsqrt}). However, as one of our primary goals is to properly understand the ground state behavior, we ultimately use the studying of the general $\beta$ free energy to eventually deduce particular ground state behavior (\ref{eq:limlogpartfunsqrt}). Along the same lines, while the analysis that we present below, in principle, accounts for any $\beta$, in the interest of easing the exposition, we may, on occasion, neglect some of the terms that in the ground state regime play no significant role.

\subsection{A replica approach to bilinear Hamiltonian based models}
\label{sec:repsqrtposhop}

In this section we  present a replica framework that can be used for studying mathematical properties of $f_{sq}(\beta)$ from (\ref{eq:logpartfunsqrt}). To that end, we start by using the standard replica ideas \cite{}
\begin{equation}
\mE_G \log(Z_{sq}(\beta,G))=\lim_{n_r\rightarrow 0}\frac{\log (\mE_G(Z_{sq}(\beta,G))^{n_r})}{n_r},\label{eq:logZn0sqrt}
\end{equation}
where $n_r$ stands for what is called the number of system's replicas.
Clearly, we view the free energy in a statistical context and along the same lines assume that the elements of $G$ are i.i.d.
standard normal random variables. Also, it is clear that if one can handle $\mE_G \log(Z_{sq}(\beta,G))$ then one can handle $\lim_{n\rightarrow \infty}\mE_G f_{sq}(n,\beta,G)$ and even more obviously $\lim_{\beta,n\rightarrow \infty}\mE_G f_{sq}(n,\beta,G)$. Hence, we focus on the right hand side of (\ref{eq:logZn0sqrt}) and start with the standard replica arguments that are typically used to handle $E_HZ_{sq}(\beta,H)^{n_r}$ (a possible disagreement with the axioms of mathematics that may happen below is considered as expected within the replica theory framework).  We first write the replicated form of partition function in the following expanded form
\begin{eqnarray}
\mE_G Z_{sq}(\beta,G)^{n_r} & = & \mE_G\left (\sum_{\x\in \cX}\sum_{\y\in\cY} e^{\beta\y^TH\x}\right )^{n_r}\nonumber \\
& = & \mE_G\left (\sum_{X\in \cX_1\times \cX_2\times \dots \times \cX_{n_r}}
\sum_{Y\in \cY_1\times \cY_2\times \dots \times \cY_{n_r}}
e^{\beta \tr(Y^THX)}\right )\nonumber \\
& = & \left (\sum_{X\in \cX_1\times \cX_2\times \dots \times \cX_{n_r}}
\sum_{Y\in \cY_1\times \cY_2\times \dots \times \cY_{n_r}}
\mE_G e^{\beta \tr(Y^THX)}\right )\nonumber \\
& = & \left (\sum_{X\in \cX_1\times \cX_2\times \dots \times \cX_{n_r}}
\sum_{Y\in \cY_1\times \cY_2\times \dots \times \cY_{n_r}}
e^{\frac{\beta^2}{2}\tr(Y^TYX^TX)}\right ),\label{eq:rep1sqrt}
\end{eqnarray}
where $\cY_i=\cY,1\leq i\leq n_r$, and $\cX_i=\cX,1\leq i\leq n_r$. The following constraint matrices $P$ and $Q$ play one of the key roles in the entire replica formalism
\begin{eqnarray}
P & = & X^TX\nonumber \\
Q & = & Y^TY.\label{eq:defQYsqrt}
\end{eqnarray}
To ensure that the above constraints remain in place over all $X$, we write
\begin{equation}
\sum_{X\in \cX_1\times \cX_2\times \dots \times \cX_{n_r}}
\delta(P-X^T X)  =  \frac{1}{2\pi^{n_r^2}}
\int d\Lambda e^{\tr(\Lambda P)}e^{n\log\left (\sum_{\rx\in \rcX}e^{-\tr(\Lambda \rx(\rx)^T)}\right )},\label{eq:defQconssqrt}
\end{equation}
where $\rcX=\cX^{n_r}$. Analogous strategy can be applied for $Y$
\begin{eqnarray}
\sum_{Y\in \cY_1\times \cY_2\times \dots \times \cY_{n_r}}
\delta(Q-Y^T Y) & = & \frac{1}{2\pi^{n_r^2}}\sum_{Y\in \cY_1\times \cY_2\times \dots \times \cY_{n_r}}
\int d\Gamma e^{\tr(\Gamma Q)-\tr(\Gamma Y^TY)}\nonumber \\
& = & \frac{1}{2\pi^{n_r^2}}
\int d\Gamma e^{\tr(\Gamma Q)}\sum_{Y\in \cY_1\times \cY_2\times \dots \times \cY_{n_r}}e^{-tr(\Gamma Y^TY)}\nonumber \\
& = & \frac{1}{2\pi^{n_r^2}}
\int d\Gamma e^{\tr(\Gamma Q)}e^{n\log\left (\sum_{Y\in \cY_1\times \cY_2\times \dots \times \cY_{n_r}}e^{-tr(\Gamma Y^TY)}\right )^{\frac{1}{n}}},\nonumber\\\label{eq:defPconssqrt}
\end{eqnarray}
Connecting (\ref{eq:rep1sqrt}), (\ref{eq:defQconssqrt}), and (\ref{eq:defPconssqrt}) we obtain
\begin{align}
\mE_G Z_{sq}(\beta,H)^{n_r}
 & = \left (\sum_{X\in \cX_1\times \cX_2\times \dots \times \cX_{n_r}}
\sum_{Y\in \cY_1\times \cY_2\times \dots \times \cY_{n_r}}
e^{\frac{\beta^2}{2}\tr(Y^TYX^TX)}\right )\nonumber \\
& = \left (\frac{1}{2\pi^{n_r^2}}\right )^2
\int d\Lambda e^{\frac{\beta^2}{2}\tr(QP)} e^{\tr(\Lambda P+\Gamma Q)}\nonumber\\
&\quad \times
e^{n\log\left (\sum_{X\in \cX_1\times \cX_2\times \dots \times \cX_{n_r}}e^{-\tr(\Lambda X^TX)}\right )^{\frac{1}{n}}}
e^{n\log\left (\sum_{Y\in \cY_1\times \cY_2\times \dots \times \cY_{n_r}}e^{-\tr(\Gamma Y^TY)}\right )^{\frac{1}{n}}}\nonumber \\
& =   \left (\frac{1}{2\pi^{n_r^2}}\right )^2
\int d\Lambda e^{ns_{sq}(\beta,\Lambda,\Gamma,P,Q)},\label{eq:rep2sqrt}
\end{align}
where
\begin{equation}
s_{sq}(\beta,\Lambda,\Gamma,P,Q)  =\frac{\beta^2}{2n}\tr(QP)+\frac{\tr\lp \Lambda P\rp}{n}+\frac{\tr\lp  \Gamma Q\rp}{n}+
\log\left (\sum_{X \in S_X}e^{-\tr(\Gamma X^TX)}\right )^{\frac{1}{n}}
+\log\left (\sum_{Y \in S_Y}e^{-\tr(\Gamma Y^TY)}\right )^{\frac{1}{n}}, \label{eq:defssaddlesqrt}
\end{equation}
and $S_Y=\cY_1\times \cY_2\times \dots \times \cY_{n_r}$ and $S_X=\cX_1\times \cX_2\times \dots \times \cX_{n_r}$. One then utilizes the saddle point method which requires that the conditions $\frac{ds_{sq, saddle}(\beta,\Lambda,\Gamma,P,Q)}{dP}=0$ and
$\frac{ds_{sq}(\beta,\Lambda,\Gamma,P,Q)}{dQ}=0$ are satisfied, i.e.,
\begin{align}
\frac{ds_{sq}(\beta,\Lambda,\Gamma,P,Q)}{dP} & =\frac{\beta^2}{2n} Q+\frac{\Lambda}{n}=0\nonumber \\
\frac{ds_{sq}(\beta,\Lambda,\Gamma,P,Q)}{dQ} & =\frac{\beta^2}{2n} P+\frac{\Gamma}{n}=0.\label{eq:saddlelamsqrt}
\end{align}
From (\ref{eq:saddlelamsqrt}) one then easily obtains
\begin{eqnarray}
\Lambda_{saddle} & = & -\frac{\beta^2}{2}Q\nonumber\\
\Gamma_{saddle} & = & -\frac{\beta^2}{2}P.\label{eq:saddlelam1sqrt}
\end{eqnarray}
Moreover, connecting (\ref{eq:defssaddlesqrt}) and (\ref{eq:saddlelam1sqrt}), we find
\begin{equation}
s_{sq}(\beta,P,Q)  = -\frac{\beta^2}{2n}\tr(QP)
 +\log\left (\sum_{X \in S_X}e^{\frac{\beta^2}{2}\tr(Q X^TX)}\right )^{\frac{1}{n}}
 +\log\left (\sum_{Y \in S_Y}e^{\frac{\beta^2}{2}\tr(P Y^TY)}\right )^{\frac{1}{n}}.\label{eq:defssaddle1sqrt}
\end{equation}
Combining (\ref{eq:logZn0sqrt}), (\ref{eq:rep2sqrt}), (\ref{eq:defssaddlesqrt}), and (\ref{eq:defssaddle1sqrt}), we further have
\begin{multline}
\lim_{n\rightarrow \infty}\frac{\mE_G \log(Z_{sq}(\beta,G))}{n}  = \lim_{n\rightarrow \infty}\lim_{n_r\rightarrow 0}\frac{\log (\mE_G(Z_{sq}(\beta,G))^{n_r})}{nn_r}
=\lim_{n\rightarrow \infty}\lim_{n_r\rightarrow 0} \frac{s_{saddle}(\beta,P,Q)}{n_r}
\\ \hspace{-.0in}=\lim_{n\rightarrow \infty} \lim_{n_r\rightarrow 0}\frac{-\frac{\beta^2}{2n}\tr(QP)
+\log\left (\sum_{X \in S_X}e^{\frac{\beta^2}{2}\tr(Q X^TX)}\right )^{\frac{1}{n}}
+\log\left (\sum_{Y \in S_Y}e^{\frac{\beta^2}{2}\tr(P Y^TY)}\right )^{\frac{1}{n}}}{n_r},\\\label{eq:ElogZ1sqrt}
\end{multline}
where  $P$ and $Q$ are presumably the saddle points of $s_{sq}(\beta,P,Q)$. Main difficulties almost completely disappeared.  However, the remaining one, regarding the choice of the saddle points $P$ and $Q$, is quite a problem on its own even within the replica theory. In what follows, we assume that $P$ and $Q$ have the same structure as does the corresponding matrix in the SK model. As is now well known, such a structure follows the so-called Parisi ansatz introduced in a sequence of Parisi's breakthrough papers \cite{Par79,Par80,Par83,Parisi80} and proven as fully rigorous for the SK model (see, e.g., \cite{Guerra03,Tal06,Pan13,Pan13a,Pan10,Pan10a}).

Since the Parisi ansatz in general assumes a potentially infinitely long hierarchial structure, we present in full generality its first level, sketch the second, and use these first two levels to later on deduce the final results for all other levels. Before proceeding further though, we should point out an alternative way of deriving the results that we present below. Namely, one can eliminate one of the matrices $P$ or $Q$ by taking the derivative with respect to the other. That effectively leaves only one of the matrices to be parameterized, which then might seem as a preferable route. On the other hand, the presentation that we choose below, is, in our view, a more natural way to uncover what really happens. After finishing the presentation, it will be rather clear  how one could proceed with the elimination of one of the matrices $P$ or $Q$ (i.e. all ingredients needed for such a direction will be obtained through the analysis that we present below).

\subsubsection{A $1$-rsb scheme based on Parisi ansatz}
\label{sec:1rsbsqrtposhop}

Along the lines of what we just discussed above, we start by observing that within the Parisi ansatz, one has that, on the first level of the symmetry breaking (1-rsb), $P$ and $Q$ have the following structure
\begin{align}
P^{(1)}=(1-p_1)I_{n_r}+I_{\frac{n_r}{m_1}}\otimes (p_1-p_0)U_{m_1\times m_1}+ p_0 U_{n_r\times n_r}\nonumber \\
Q^{(1)}=(1-\bq_1)I_{n_r}+I_{\frac{n_r}{m_1}}\otimes (\bq_1-\bq_0)U_{m_1\times m_1}+ \bq_0 U_{n_r\times n_r},\label{eq:Qsaddle1rsbsqrt}
\end{align}
where $I$ is the identity matrix  and $U$ is the matrix of all ones (the dimensions of these matrices are  specified in their  subscripts). To facilitate the exposition, we assumed in (\ref{eq:Qsaddle1rsbsqrt}) that the elements of $\cX$ and $\cY$ have unit norms (everything that follows can be repeated line by line without such an assumption but the writing is substantially less elegant). Even if one accepts that
$P$ and $Q$ should share the same type of internal structure, it is not a priori clear why the dimensions of the submatrices that make up such structure should be the same for both $P^{(1)}$ and $Q^{(1)}$. We, nonetheless, find it useful to make such an assumption.

To handle (\ref{eq:defssaddle1sqrt}) one then needs the eigenvalues of $ P^{(1)}$ and $Q^{(1)}$. It is not that hard to verify that the eigenvalues of $P^{(1)}$ are as follows,
\begin{eqnarray}
n_r-\frac{n_r}{m_1} \quad \mbox{eigenvalues equal to}: & & \lambda_{1}^{(1)}=(1-p_1)\nonumber \\
\frac{n_r}{m_1}-1 \quad \mbox{eigenvalues equal to}: & & \lambda_{2}^{(1)}=(1-p_1)+m_1(p_1-p_0)\nonumber \\
1 \quad \mbox{eigenvalue equal to}: & & \lambda_{3}^{(1)}=(1-p_1)+ m_1(p_1-p_0)+n_rp_0,\label{eq:eigQsaddle1rsbsqrt}
\end{eqnarray}
and
that the eigenvalues of $Q^{(1)}$ are as follows,
\begin{eqnarray}
n_r-\frac{n_r}{m_1} \quad \mbox{eigenvalues equal to}: & & \bar{\lambda}_{1}^{(1)}=(1-\bq_1)\nonumber \\
\frac{n_r}{m_1}-1 \quad \mbox{eigenvalues equal to}: & & \bar{\lambda}_{2}^{(1)}=(1-\bq_1)+m_1(\bq_1-\bq_0)\nonumber \\
1 \quad \mbox{eigenvalue equal to}: & & \bar{\lambda}_{3}^{(1)}=(1-\bq_1)+ m_1(\bq_1-\bq_0)+n_r\bq_0.\label{eq:eigbarQsaddle1rsbsqrt}
\end{eqnarray}
 Now, it is possible to proceed with the computation of $s_{sq}(\beta,P^{(1)},Q^{(1)})$.  First, we recognize the following
\begin{eqnarray}
\lim_{n_r\rightarrow 0,n\rightarrow \infty} \frac{s_{sq}(\beta,P^{(1)},Q^{(1)})}{n_r}
& = &  \lim_{n_r\rightarrow 0,n\rightarrow \infty}\frac{1}{n_r}
\Bigg( \Bigg.   -\frac{\beta^2}{2n}\tr(Q^{(1)}P^{(1)})
+  \log\left (\sum_{X \in S_X}e^{\frac{\beta^2}{2}\tr(Q^{(1)} X^TX)}\right )^{\frac{1}{n}} \nonumber \\
& & +  \log\left (\sum_{Y \in S_Y}e^{\frac{\beta^2}{2}\tr(P^{(1)} Y^TY)}\right )^{\frac{1}{n}}  \Bigg.\Bigg) \nonumber \\
& = &  I_{Q,P}^{(1)}+I_{Q}^{(1)}+I_{P}^{(1)}, \label{eq:defssaddleIssqrt}
\end{eqnarray}
where
\begin{eqnarray}
 I_{Q,P}^{(1)} & = & \lim_{n_r\rightarrow 0,n\rightarrow \infty}\frac{-\frac{\beta^2}{2n}\tr(Q^{(1)}P^{(1)})}{n_r}\nonumber \\
I_{Q}^{(1)}& = & \lim_{n_r\rightarrow 0,n\rightarrow \infty}\frac{\log\left (\sum_{X \in S_X}e^{\frac{\beta^2}{2}\tr(Q^{(1)} X^TX)}\right )^{\frac{1}{n}}}{n_r}\nonumber \\
 I_{P}^{(1)}& = & \lim_{n_r\rightarrow 0,n\rightarrow \infty}\frac{\log\left (\sum_{Y \in S_Y}e^{\frac{\beta^2}{2}\tr(P^{(1)} Y^TY)}\right )^{\frac{1}{n}}}{n_r}. \label{eq:saddleIssqrt}
\end{eqnarray}
We below split the computation into three steps where we separately compute each of $I_{Q,P}^{(1)}$, $I_{P}^{(1)}$, and $I_{Q}^{(1)}$.

\noindent \emph{\underline{1) Determining $I_{Q,P}^{(1)}$}}
\vspace{.1in}

\noindent We directly calculate $\tr(Q^{(1)}P^{(1)})$ by relying on the above determined eigenvalues
\begin{align}
\hspace{-.0in}I_{Q,P}^{(1)} & =\lim_{n_r\rightarrow 0,n\rightarrow \infty} -\frac{\beta^2}{2nn_r}
\Bigg(\Bigg. \left (n_r-\frac{n_r}{m_1}\right )(1-\bq_1)(1-p_1) \nonumber \\
& \quad
+\left (\frac{n_r}{m_1}-1\right )\left ( 1-\bq_1+m_1(\bq_1-\bq_0)\right )\left ( 1-p_1+m_1(p_1-p_0)\right )\nonumber \\
& \quad +\left ( 1-\bq_1+m_1(\bq_1-\bq_0)+n_r\bq_0\right )\left ( 1-p_1+m_1(p_1-p_0)+n_rp_0\right ) \Bigg.\Bigg)\nonumber \\
& =\lim_{n_r\rightarrow 0,n\rightarrow \infty} -\frac{\beta^2}{2nn_r}(n_r(1-\bq_1)(1-p_1)
+n_r(\bq_1-\bq_0)(1-p_1)+n_r(p_1-p_0)(1-\bq_1)  \nonumber \\
& \quad +n_rm_1(\bq_1-\bq_0)(p_1-p_0)
+n_rp_0\left ( 1-\bq_1+m_1(\bq_1-\bq_0)\right )+n_r\bq_0\left ( 1-p_1+m_1(p_1-p_0)\right ) +n_r^2\bq_0p_0)\nonumber \\
& =\lim_{n_r\rightarrow 0,n\rightarrow \infty} -\frac{\beta^2}{2nn_r}(n_r(1-\bq_1)(1-p_1)
+n_r\bq_1(1-p_1)+n_rp_1(1-\bq_1) +n_rm_1(\bq_1-\bq_0)(p_1-p_0)\nonumber \\
& \quad +n_rp_0m_1(\bq_1-\bq_0)+n_r\bq_0 m_1(p_1-p_0) )\nonumber \\
& =\lim_{n\rightarrow \infty} -\frac{\beta^2}{2n}(1-\bq_1p_1 +m_1(\bq_1p_1-\bq_0p_0) ).\label{eq:finalsaddleI1sqrt}
\end{align}

\vspace{.1in}

\noindent \emph{\underline{2) Determining $I_{P}^{(1)}$}}
\vspace{.1in}

\noindent  We start with
\begin{equation}
 I_{P}^{(1)}= \lim_{n_r\rightarrow 0,n\rightarrow \infty}\frac{\log\left (\sum_{Y \in S_Y}e^{\frac{\beta^2}{2}\tr(P^{(1)} Y^TY)}\right )^{\frac{1}{n}}}{n_r}, \label{eq:saddleI3sqrt}
\end{equation}
and introduce
\begin{equation}
\hspace{-0in}g_{tot}^{(1)}=\frac{\beta^2}{2}\tr(P^{(1)}Y^TY)=\frac{n_r\beta^2(1-p_1)}{2}
+\frac{\beta^2(p_1-p_0)\sum_{i=1}^{\frac{n_r}{m_1}}\|\1^TY^{(i)}\|_2^2}
{2}
+\frac{\beta^2 p_0 \|\1^TY^T\|_2^2}{2},\label{eq:g3totsqrt}
\end{equation}
where
\begin{equation}
Y^{(i)}=Y^T_{(i-1)m_1+1:im_1,1:n},1\leq i\leq \frac{n_r}{m_1}.\label{eq:defYisqrt}
\end{equation}
Moreover, we write
\begin{equation}
\hspace{-0in}g_{tot}^{(1)}=\frac{\beta^2}{2}\tr(P^{(1)}Y^TY)=\frac{n_r\beta^2(1-p_1)}{2}
+\frac{\alpha_{1}^2\sum_{i=1}^{\frac{n_r}{m_1}}\|\1^TY^{(i)}\|_2^2}
{2}
+\frac{\alpha_{2}^2 \|\1^TY^T\|_2^2}{2},\label{eq:g3tot1sqrt}
\end{equation}
where
\begin{eqnarray}
\alpha_{1} & = & \beta\sqrt{p_1-p_0}\nonumber \\
\alpha_{2} & = & \beta\sqrt{p_0},\label{eq:alphabar12sqrt}
\end{eqnarray}
and, both here and throughout the paper, $\1$ is a column vector of all ones of appropriate dimension. Combining (\ref{eq:saddleI3sqrt}) and (\ref{eq:g3tot1sqrt}), we have
\begin{align}
\hspace{-.0in}I_{P}^{(1)} & = \lim_{n_r\rightarrow 0,n\rightarrow \infty}\frac{\log\left (\sum_{Y \in S_Y}e^{\frac{n\beta^2}{2}\tr(P^{(1)} Y^TY)}\right )^{\frac{1}{n}}}{nn_r}
=\lim_{n_r\rightarrow 0,n\rightarrow \infty}\frac{\log\left (\sum_{Y\in S_Y}e^{g_{tot}^{(1)}}\right )}{nn_r}\nonumber \\
& =  \lim_{n_r\rightarrow 0,n\rightarrow \infty}
\frac{\beta^2(1-p_1)}{2n}+
\frac{\log\left (\sum_{Y\in S_Y}e^{\frac{\alpha_{1}^2\sum_{i=1}^{\frac{n_r}{m_1}}\|\1^TY^{(i)}\|_2^2}
{2}}\int_{\z^{(2)}}e^{\alpha_{2} \sum_{i=1}^{\frac{n_r}{m_1}}\1^TY^{(i)}\z^{(2)}}e^{-\frac{\|\z^{(2)}\|_2^2}{2}}d\z^{(2)}\right )}{nn_r}\nonumber \\
& = \lim_{n_r\rightarrow 0,n\rightarrow \infty}
 \frac{\beta^2(1-p_1)}{2n}+
\frac{\log\left (\int_{\z^{(2)}}\sum_{Y\in S_Y}e^{\frac{\alpha_{1}^2\sum_{i=1}^{\frac{n_r}{m_1}}\|\1^TY^{(i)}\|_2^2}
{2}}e^{\alpha_{2}\sum_{i=1}^{\frac{n_r}{m_1}}\1^TY^{(i)}\z^{(2)}}e^{-\frac{\|\z^{(2)}\|_2^2}{2}}d\z^{(2)}\right )}{nn_r}\nonumber \\
& =  \lim_{n_r\rightarrow 0,n\rightarrow \infty}
\frac{\beta^2(1-p_1)}{2n}+
\frac{\log\left (\int_{\z^{(2)}}\left (\sum_{Y^{(i)}\in S_Y^{(m_1)}}e^{\frac{\alpha_{1}^2\|\1^TY^{(i)}\|_2^2}
{2}}e^{\alpha_{2}\1^TY^{(i)}\z^{(2)}}\right )^{\frac{n_r}{m_1}}e^{-\frac{\|\z^{(2)}\|_2^2}{2}}d\z^{(2)}\right )}{nn_r}, \nonumber \\\label{eq:saddleI3gtotsqrt}
\end{align}
where $S_Y^{(m_1)}=\underbrace{\cY\times \cY \times \dots \times \cY}_{m_1\hspace{.03in} times}$. A bit of further algebraic transformations gives
\begin{align}
\hspace{-.0in}I_{P}^{(1)} & =  \lim_{n_r\rightarrow 0,n\rightarrow \infty}
\frac{\beta^2(1-p_1)}{2n}+
\frac{\log\left (\int_{\z^{(2)}}e^{\log\left (\sum_{Y^{(i)}\in S_Y^{(m_1)}}e^{\frac{\alpha_{1}^2\|\1^TY^{(i)}\|_2^2}
{2}}e^{\alpha_{2} \1^TY^{(i)}\z^{(2)}}\right )^{\frac{n_r}{m_1}}}e^{-\frac{\|\z^{(2)}\|_2^2}{2}}d\z^{(2)}\right )}{nn_r}\nonumber \\
& = \lim_{n_r\rightarrow 0,n\rightarrow \infty}
\frac{\beta^2(1-p_1)}{2n}+
\frac{\log\left (\int_{\z^{(2)}}e^{\frac{n_r}{m_1}\log\left (\sum_{Y^{(i)}\in S_Y^{(m_1)}}e^{\frac{\alpha_{1}^2\|\1^TY^{(i)}\|_2^2}
{2}}e^{\alpha_{2}\1^TY^{(i)}\z^{(2)}}\right )}e^{-\frac{\|\z^{(2)}\|_2^2}{2}}d\z^{(2)}\right )}{nn_r}.\nonumber \\ \label{eq:saddleI3gtot1sqrt}
\end{align}
As $n_r\rightarrow 0$ we also have
\begin{eqnarray}
\hspace{-.0in}I_{P}^{(1)} & = &  \lim_{n\rightarrow \infty} \frac{\beta^2(1-p_1)}{2n} \nonumber \\
& & +\lim_{n_r\rightarrow 0,n\rightarrow \infty}\frac{\log\left (\int_{\z^{(2)}}e^{\frac{n_r}{m_1}\log\left (\sum_{Y^{(i)}\in S_Y^{(m_1)}}e^{\frac{\alpha_{1}^2\|\1^TY^{(i)}\|_2^2}
{2}}e^{\alpha_{2}\1^TY^{(i)}\z^{(2)}}\right )}e^{-\frac{\|\z^{(2)}\|_2^2}{2}}d\z^{(2)}\right )}{nn_r}  \nonumber \\
\hspace{-.0in} & = &  \lim_{n\rightarrow \infty} \frac{\beta^2(1-p_1)}{2n} \nonumber \\
& & +\lim_{n_r\rightarrow 0,n\rightarrow \infty}\frac{\log\left (\int_{\z^{(2)}}\left (1+\frac{n_r}{m_1}\log\left (\sum_{Y^{(i)}\in S_Y^{(m_1)}}e^{\frac{\alpha_{1}^2\|\1^TY^{(i)}\|_2^2}
{2}}e^{\alpha_{2}\1^TY^{(i)}\z^{(2)}}\right )\right )e^{-\frac{\|\z^{(2)}\|_2^2}{2}}d\z^{(2)}\right )}{nn_r} \nonumber \\
\hspace{-.0in} & =  &  \lim_{n\rightarrow \infty} \frac{\beta^2(1-p_1)}{2n} \nonumber \\
& & +\lim_{n_r\rightarrow 0,n\rightarrow \infty}\frac{\log\left (1+\frac{n_r}{m_1}\int_{\z^{(2)}}\log\left (\sum_{Y^{(i)}\in S_Y^{(m_1)}}e^{\frac{\alpha_{1}^2\|\1^TY^{(i)}\|_2^2}
{2}}e^{\alpha_{2}\1^TY^{(i)}\z^{(2)}}\right )e^{-\frac{\|\z^{(2)}\|_2^2}{2}}d\z^{(2)}\right )}{nn_r} \nonumber  \\
 & = &  \lim_{n\rightarrow \infty}\frac{\beta^2(1-p_1)}{2n}+\lim_{n\rightarrow \infty}\frac{\frac{1}{m_1}\int_{\z^{(2)}}\log\left (\sum_{Y^{(i)}\in S_Y^{(m_1)}}e^{\frac{\alpha_{1}^2\|\1^TY^{(i)}\|_2^2}
{2}}e^{\alpha_{2}\1^TY^{(i)}\z^{(2)}}\right )e^{-\frac{\|\z^{(2)}\|_2^2}{2}}d\z^{(2)}}{n}. \nonumber \\
\label{eq:saddleI3gtot2sqrt}
\end{eqnarray}
Moreover,
\begin{align}
  I_{P}^{(1)}   & =  \lim_{n\rightarrow \infty} \frac{\beta^2(1-p_1)}{2n} \nonumber \\
 & \quad +\lim_{n\rightarrow \infty}\frac{\frac{1}{m_1}\int_{\z^{(2)}}\log\left (\sum_{Y^{(i)}\in S_Y^{(m_1)}}\int_{\z^{(1)}}e^{\alpha_{1}\1^TY^{(i)}\z^{(1)}
}e^{-\frac{\|\z^{(1)}\|_2^2}{2}}d\z^{(1)}e^{\alpha_{2}\1^TY^{(i)}\z^{(2)}}\right )e^{-\frac{\|\z^{(2)}\|_2^2}{2}}d\z^{(2)}}{n} \nonumber \\
\hspace{-.0in} & =  \lim_{n_r\rightarrow 0,n\rightarrow \infty} \frac{\beta^2(1-p_1)}{2n}  \nonumber \\
 & \quad +\lim_{n\rightarrow \infty}\frac{\frac{1}{m_1}\int_{\z^{(2)}}\log\left (\int_{\z^{(1)}}\sum_{Y^{(i)}\in S_Y^{(m_1)}}e^{\alpha_{1}\1^TY^{(i)}\z^{(1)}
}e^{\alpha_{2} \1^TY^{(i)}\z^{(2)}}e^{-\frac{\|\z^{(1)}\|_2^2}{2}}d\z^{(1)}\right )e^{-\frac{\|\z^{(2)}\|_2^2}{2}}d\z^{(2)}}{n} \nonumber \\
\hspace{-.0in} & =  \lim_{n\rightarrow \infty} \frac{\beta^2(1-p_1)}{2n}  \nonumber \\
 & \quad +\lim_{n\rightarrow \infty}\frac{\frac{1}{m_1}\int_{\z^{(2)}}\log\left (\int_{\z^{(1)}}\left (\sum_{\y\in \cY}e^{\y^T\left ( \alpha_{1} \z^{(1)}+\alpha_{2} \z^{(2)}\right )}\right )^{m_1}e^{-\frac{\|\z^{(1)}\|_2^2}{2}}d\z^{(1)}\right )e^{-\frac{\|\z^{(2)}\|_2^2}{2}}d\z^{(2)}}{n}. \nonumber \\ \label{eq:saddleI3gtot3sqrt}
\end{align}
It is not that difficult to see that $I_{P}^{(1)}$ and (\ref{eq:saddleI3gtot3sqrt}) can, in fact, be rewritten in an even more compact form
\begin{equation}
  I_{P}^{(1)}= \lim_{n\rightarrow \infty}\frac{E_{\z^{(2)}}\log\left (E_{\z^{(1)}}\left ( E_{\z^{(0)}}\left (\sum_{\y\in\cY}e^{\y^T ( \alpha_{0}\z^{(0)}+\alpha_{1}\z^{(1)}+\alpha_{2} \z^{(2)})}
\right )\right )^{m_1}\right )}{m_1n},\label{eq:finalsaddleI3sqrtcomp}
\end{equation}
where
\begin{eqnarray}
\alpha_{0} & = & \beta\sqrt{1-p_1}.\label{eq:alpha0sqrt}
\end{eqnarray}

\vspace{.1in}

\noindent \emph{\underline{3) Determining $I_{Q}^{(1)}$}}
\vspace{.1in}

\noindent Since $I_{Q}^{(1)}$ is structurally identical to $I_{P}^{(1)}$, following step-by-step the above derivation, we obtain, as a complete analogue to (\ref{eq:saddleI3gtot3sqrt}),
\begin{align}
  I_{Q}^{(1)} \hspace{-.0in} & =  \lim_{n\rightarrow \infty} \frac{\beta^2(1-\bq_1)}{2n}  \nonumber \\
 & \quad +\lim_{n\rightarrow \infty}\frac{\frac{1}{m_1}\int_{\z^{(2)}}\log\left (\int_{\z^{(1)}}\left (\sum_{\x\in \cX} e^{\x^T\left ( \bar{\alpha}_{1}\sqrt{n}\z^{(1)}+\bar{\alpha}_{2} \sqrt{n}\z^{(2)}\right )}\right )^{m_1}e^{-\frac{\|\z^{(1)}\|_2^2}{2}}d\z^{(1)}\right )e^{-\frac{\|\z^{(2)}\|_2^2}{2}}d\z^{(2)}}{n}, \nonumber \\ \label{eq:saddleI2gtot3sqrt}
\end{align}
or in a more compact form
\begin{equation}
  I_{P}^{(1)}= \lim_{n\rightarrow \infty}\frac{E_{\z^{(2)}}\log\left (E_{\z^{(1)}}\left ( E_{\z^{(0)}}\left (\sum_{\x\in\cX}e^{\x^T ( \bar{\alpha}_{0}\sqrt{n}\z^{(0)}+\bar{\alpha}_{1}\sqrt{n}\z^{(1)}+\bar{\alpha}_{2} \sqrt{n}\z^{(2)})}
\right )\right )^{m_1}\right )}{m_1n},\label{eq:finalsaddleI2sqrtcomp}
\end{equation}
where
\begin{eqnarray}
\bar{\alpha}_{0} & = & \beta\sqrt{1-\bq_1}\nonumber \\
\bar{\alpha}_{1} & = & \beta\sqrt{\bq_1-\bq_0}\nonumber \\
\bar{\alpha}_{2} & = & \beta\sqrt{\bq_0}.\label{eq:baralpha123sqrt}
\end{eqnarray}
From (\ref{eq:limlogpartfunsqrt}), we have for $\beta=\bar{\beta}\sqrt{n}$ (where $\bar{\beta}$ does \emph{not} change as $n$ grows)
\begin{eqnarray}
 f_{sq}(\beta) & = &
\lim_{n\rightarrow\infty}\frac{\mE_G \log{(Z_{sq}(\beta,G)})}{\beta \sqrt{n}}
=\lim_{n\rightarrow\infty}\frac{\mE_G \log{(Z_{sq}(\beta,G)})}{\frac{\beta}{\sqrt{n}} n}.\label{eq:limlogpartfunsqrtaa00}
\end{eqnarray}
A combination of (\ref{eq:ElogZ1sqrt}), (\ref{eq:defssaddleIssqrt}), (\ref{eq:finalsaddleI1sqrt}), (\ref{eq:finalsaddleI3sqrtcomp}), and (\ref{eq:finalsaddleI2sqrtcomp}), gives $\bar{f}_{sq}^{(1)}(\beta)$ as a $1$-rsb estimate of $f_{sq}(\beta)$,
\begin{eqnarray}
\bar{f}_{sq}^{(1)}(\beta) &  = & -\lim_{n\rightarrow \infty} \frac{\beta^2}{2\frac{\beta}{\sqrt{n}}n}(1-\bq_1p_1 +m_1(\bq_1p_1-\bq_0p_0) ) \nonumber \\
& & +\lim_{n\rightarrow \infty}\frac{E_{\z^{(2)}}\log\left (E_{\z^{(1)}}\left ( E_{\z^{(0)}}\left (\sum_{\x\in\cX}e^{\x^T ( \bar{\alpha}_{0}\z^{(0)}+\bar{\alpha}_{1}\z^{(1)}+\bar{\alpha}_{2}\z^{(2)})}
\right )\right )^{m_1}\right )}{\frac{\beta}{\sqrt{n}} m_1n} \nonumber \\
& & +\lim_{n\rightarrow \infty}\frac{E_{\z^{(2)}}\log\left (E_{\z^{(1)}}\left ( E_{\z^{(0)}}\left (\sum_{\y\in\cY}e^{\y^T ( \alpha_{0}\z^{(0)}+\alpha_{1}\z^{(1)}+\alpha_{2}\z^{(2)})}
\right )\right )^{m_1}\right )}{\frac{\beta}{\sqrt{n}} m_1n}.\label{eq:finalfreeenergycompsqrt}
\end{eqnarray}
In fact, one can go a step further and obtain the following
\begin{eqnarray}
\bar{f}_{sq}^{(1)}(\beta) & = &
-\lim_{n\rightarrow \infty}
\frac{E_{z_2}\log \left ( E_{z_1}\left (E_{z_0}e^{(\tilde{\alpha}_{0}  z_0+\tilde{\alpha}_{1}  z_1+\tilde{\alpha}_{2}  z_2)}\right )^{m_1}\right )}{\frac{\beta}{\sqrt{n}}m_1n}
\\
& & +\lim_{n\rightarrow \infty}\frac{E_{\z^{(2)}}\log\left (E_{\z^{(1)}}\left ( E_{\z^{(0)}}\left (\sum_{\x\in\cX}e^{\x^T ( \bar{\alpha}_{0}\z^{(0)}+\bar{\alpha}_{1}\z^{(1)}+\bar{\alpha}_{2}\z^{(2)})}
\right )\right )^{m_1}\right )}{\frac{\beta}{\sqrt{n}} m_1n} \nonumber \\
& & +\lim_{n\rightarrow \infty}\frac{E_{\z^{(2)}}\log\left (E_{\z^{(1)}}\left ( E_{\z^{(0)}}\left (\sum_{\y\in\cY}e^{\y^T ( \alpha_{0}\z^{(0)}+\alpha_{1}\z^{(1)}+\alpha_{2}\z^{(2)})}
\right )\right )^{m_1}\right )}{\frac{\beta}{\sqrt{n}} m_1n},
\label{eq:finalfreeenergycomp1sqrt}
\end{eqnarray}
where
\begin{eqnarray}
\tilde{\alpha}_{0} & = & \beta\sqrt{1-\bq_1p_1}\nonumber \\
\tilde{\alpha}_{1} & = & \beta\sqrt{\bq_1p_1-\bq_0p_0}\nonumber \\
\tilde{\alpha}_{2} & = & \beta\sqrt{\bq_0p_0}.\label{eq:tildealpha123sqrt}
\end{eqnarray}
The following then summarizes what we presented in this subsection.
\begin{observation} (\textbf{1-rsb})
Let the thermodynamic limit of the averaged free energy of generic bilinear Hamiltonian based model, $f_{sq}(\beta)$, be as defined in (\ref{eq:logpartfunsqrt}). Let matrices $P^{(1)}$
and $Q^{(1)}$ be as defined in (\ref{eq:Qsaddle1rsbsqrt}). Further let $\alpha$, $\bar{\alpha}$, and $\tilde{\alpha}$ sequences be as defined in
(\ref{eq:alphabar12sqrt}), (\ref{eq:alpha0sqrt}), (\ref{eq:baralpha123sqrt}), and (\ref{eq:tildealpha123sqrt})
Then, for appropriately selected parameters in the definition of $P^{(1)}$ and $Q^{(1)}$, one has that the  $1$-rsb estimate for $f_{sq}(\beta)$, $\bar{f}_{sq}^{(1)}(\beta)$, is as given in (\ref{eq:finalfreeenergycomp1sqrt}).
\label{obs:1rsbsqrt}
\end{observation}
Observation \ref{obs:1rsbsqrt} is then sufficient to analyze the average free energy and its  ground state special case on the first level of symmetry breaking. In the following subsection, we look at what happens if one continues to break the structure of $P$ and $Q$ according to the Parisi ansatz.

\subsubsection{A $2$-rsb scheme based on Parisi ansatz}
\label{sec:2rsbsqrtposhop}

Move from the $1$-rsb to $2$-rsb structure of $P$ and $Q$ makes the underlying derivations significantly more cumbersome. To facilitate the presentation and overall clarity, we skip a majority of repetitive details and present only the main new steps and the final results. The overall strategy, of course,  is literally the same as the one presented in Section \ref{sec:1rsbsqrtposhop}.

Along the lines of what was done in Section \ref{sec:1rsbsqrtposhop}, we start by observing that, within the Parisi ansatz, on the second level of the symmetry breaking, $P$ and $Q$ have the following structure
\begin{eqnarray}
P^{(2)}=(1-p_2)I_{n_r}+I_{\frac{n_r}{m_2}}\otimes (p_2-p_1)U_{m_2\times m_2}+
I_{\frac{n_r}{m_1}}\otimes (p_1-p_0)U_{m_1\times m_1}+ p_0 U_{n_r\times n_r}\nonumber \\
Q^{(2)}=(1-\bq_2)I_{n_r}+I_{\frac{n_r}{m_2}}\otimes (\bq_2-\bq_1)U_{m_2\times m_2}+
I_{\frac{n_r}{m_1}}\otimes (\bq_1-\bq_0)U_{m_1\times m_1}+ \bq_0 U_{n_r\times n_r}\nonumber \\
,\label{eq:2rsbQsaddle2rsbsqrt}
\end{eqnarray}
where $I$  and $U$, are as in Section \ref{sec:1rsbsqrtposhop}. One then has
\begin{eqnarray}
n_r-\frac{n_r}{m_2} \quad \mbox{eigenvalues equal to}: & & \lambda_{1}^{(2)}=1-p_2\nonumber \\
\frac{n_r}{m_2}-\frac{n_r}{m_1} \quad \mbox{eigenvalues equal to}: & & \lambda_{2}^{(2)}=(1-p_2)+m_2(p_2-p_1)\nonumber \\
\frac{n_r}{m_1}-1 \quad \mbox{eigenvalues equal to}: & & \lambda_{3}^{(2)}=(1-p_2)+ m_2(p_2-p_1)+ m_1(p_1-p_0)\nonumber \\
1 \quad \mbox{eigenvalue equal to}: & & \lambda_{4}^{(2)}=(1-p_2)+ m_2(p_2-p_1)+ m_1(p_1-p_0)+n_r\beta p_0,\nonumber \\\label{eq:2rsbeigQsaddle2rsbsqrt}
\end{eqnarray}
for the eigenvalues of matrix $P^{(2)}$ and
\begin{eqnarray}
n_r-\frac{n_r}{m_2} \quad \mbox{eigenvalues equal to}: & & \bar{\lambda}_{1}^{(2)}=1-\bq_2\nonumber \\
\frac{n_r}{m_2}-\frac{n_r}{m_1} \quad \mbox{eigenvalues equal to}: & & \bar{\lambda}_{2}^{(2)}=(1-\bq_2)+m_2(\bq_2-\bq_1)\nonumber \\
\frac{n_r}{m_1}-1 \quad \mbox{eigenvalues equal to}: & & \bar{\lambda}_{3}^{(2)}=(1-\bq_2)+ m_2(\bq_2-\bq_1)+ m_1(\bq_1-\bq_0)\nonumber \\
1 \quad \mbox{eigenvalue equal to}: & & \bar{\lambda}_{4}^{(2)}=(1-\bq_2)+ m_2(\bq_2-\bq_1)+ m_1(\bq_1-\bq_0)+n_r\beta \bq_0,\nonumber \\\label{eq:2rsbeigQbarsaddle2rsbsqrt}
\end{eqnarray}
for the eigenvalues of matrix $Q^{(2)}$. Following into the footsteps of Section \ref{sec:1rsbsqrtposhop}, the computation of $s_{sq}(\beta,P^{(2)},Q^{(2)})$ is again split into three steps. First, we recognize
\begin{eqnarray}
\lim_{n_r\rightarrow 0,n\rightarrow \infty} \frac{s_{sq}(\beta,P^{(2)},Q^{(2)})}{n_r}
& = &  \lim_{n_r\rightarrow 0,n\rightarrow \infty}\frac{1}{n_r}
\Bigg( \Bigg.   -\frac{\beta^2}{2n}\tr(Q^{(2)}P^{(2)})
+  \log\left (\sum_{X \in S_X}e^{\frac{\beta^2}{2}\tr(Q^{(2)} X^TX)}\right )^{\frac{1}{n}} \nonumber \\
& & +  \log\left (\sum_{Y \in S_Y}e^{\frac{\beta^2}{2}\tr(P^{(2)} Y^TY)}\right )^{\frac{1}{n}}  \Bigg.\Bigg) \nonumber \\
& = &  I_{Q,P}^{(2)}+I_{Q}^{(2)}+I_{P}^{(2)}, \label{eq:2rsbdefssaddleIs2rsbsqrt}
\end{eqnarray}
where
\begin{eqnarray}
 I_{Q,P}^{(2)} & = & \lim_{n_r\rightarrow 0,n\rightarrow \infty}\frac{-\frac{\beta^2}{2n}\tr(Q^{(2)}P^{(2)})}{n_r}\nonumber \\
I_{Q}^{(2)}& = & \lim_{n_r\rightarrow 0,n\rightarrow \infty}\frac{\log\left (\sum_{X \in S_X}e^{\frac{\beta^2}{2}\tr(Q^{(2)} X^TX)}\right )^{\frac{1}{n}}}{n_r}\nonumber \\
 I_{P}^{(2)}& = & \lim_{n_r\rightarrow 0,n\rightarrow \infty}\frac{\log\left (\sum_{Y \in S_Y}e^{\frac{\beta^2}{2}\tr(P^{(2)} Y^TY)}\right )^{\frac{1}{n}}}{n_r}. \label{eq:2rsbsaddleIs2rsbsqrt}
\end{eqnarray}
Below, we separately compute each of $I_{Q,P}^{(2)}$, $I_{P}^{(2)}$, and $I_{Q}^{(2)}$.

\noindent \emph{\underline{1) Determining $I_{Q,P}^{(2)}$}}
\vspace{.1in}

\noindent We start with
\begin{align}
\hspace{-.0in}I_{Q,P}^{(2)} & =\lim_{n_r\rightarrow 0,n\rightarrow \infty} -\frac{\beta^2}{2nn_r}\Bigg(\Bigg. \left (n_r-\frac{n_r}{m_2}\right )(1-\bq_2)(1-p_2) \nonumber \\
& \quad +\left (\frac{n_r}{m_2}-\frac{n_r}{m_1}\right )\left ( 1-\bq_2+m_2(\bq_2-\bq_1)\right )\left ( 1-p_2+m_2(p_2-p_1)\right )\nonumber \\
& \quad +\left (\frac{n_r}{m_1}-1\right )\left ( 1-\bq_2+m_2(\bq_2-\bq_1)+m_1(\bq_1-\bq_0)\right )\left ( 1-p_2+m_2(p_2-p_1)+m_1(p_1-p_0)\right )\nonumber \\
& \quad +\left ( 1-\bq_2+m_2(\bq_2-\bq_1)+m_1(\bq_1-\bq_0)+n_r\bq_0\right )\left ( 1-p_2+m_2(p_2-p_1)+m_1(p_1-p_0)+n_rp_0\right ) \Bigg. \Bigg)\nonumber \\
& =\lim_{n\rightarrow \infty} -\frac{\beta^2}{2n}(1-\bq_2p_2 +m_2(\bq_2p_2-\bq_1p_1)+m_1(\bq_1p_1-\bq_0p_0) ).\label{eq:2rsbfinalsaddleI12rsbsqrt}
\end{align}

\noindent \emph{\underline{2) Determining $I_{P}^{(2)}$}}
\vspace{.1in}

\noindent  We start with
\begin{equation}
 I_{P}^{(2)}= \lim_{n_r\rightarrow 0,n\rightarrow \infty}\frac{\log\left (\sum_{Y \in S_Y}e^{\frac{\beta^2}{2}\tr(P^{(2)} Y^TY)}\right )^{\frac{1}{n}}}{n_r}, \label{eq:2rsbsaddleI3sqrt}
\end{equation}
and introduce
\begin{eqnarray}
\hspace{-0in}g_{tot}^{(2)}
& = & \frac{\beta^2}{2}\tr(P^{(2)}Y^TY)
 \nonumber \\
& = &
\frac{n_r\beta^2(1-p_2)}{2}
+\frac{\beta^2(p_2-p_1)\sum_{i_2=1}^{\frac{n_r}{m_2}}\|\1^TY^{(i,2)}\|_2^2}
{2}
+\frac{\beta^2(p_1-p_0)\sum_{i_1=1}^{\frac{n_r}{m_1}}\|\1^TY^{(i,1)}\|_2^2}
{2} \nonumber \\
& & +\frac{\beta^2 p_0 \|\1^TY^T\|_2^2}{2},\nonumber \\
\label{eq:2rsbg3totsqrt}
\end{eqnarray}
where
\begin{eqnarray}
Y^{(i,k)}=Y^T_{(i-1)m_k+1:im_k,1:n},1\leq i\leq \frac{n_r}{m_k}, k\in\{1,2\}.\label{eq:2rsbdefYisqrt}
\end{eqnarray}
Moreover, we can also write
\begin{equation}
\hspace{-0in}g_{tot}^{(2)}=\frac{\beta^2}{2}\tr(P^{(2)}Y^TY)=\frac{n_r\beta^2(1-p_2)}{2}
+\frac{\alpha_{1}^2\sum_{i_2=1}^{\frac{n_r}{m_1}}\|\1^TY^{(i,2)}\|_2^2}
{2}
+\frac{\alpha_{2}^2\sum_{i_1=1}^{\frac{n_r}{m_1}}\|\1^TY^{(i,1)}\|_2^2}
{2}
+\frac{\alpha_{3}^2 \|\1^TY^T\|_2^2}{2},\label{eq:2rsbg3tot1sqrt}
\end{equation}
where
\begin{eqnarray}
\alpha_{1} & = & \beta\sqrt{p_2-p_1}\nonumber \\
\alpha_{2} & = & \beta\sqrt{p_1-p_0}\nonumber \\
\alpha_{3} & = & \beta\sqrt{p_0}.\label{eq:2rsbalphabar12sqrt}
\end{eqnarray}
Combining (\ref{eq:2rsbsaddleI3sqrt}) and (\ref{eq:2rsbg3tot1sqrt}), we have
\begin{align}
\hspace{-.0in}I_{P}^{(2)} & = \lim_{n_r\rightarrow 0,n\rightarrow \infty}\frac{\log\left (\sum_{Y \in S_Y}e^{\frac{n\beta^2}{2}\tr(P^{(2)} Y^TY)}\right )^{\frac{1}{n}}}{nn_r}
=\lim_{n_r\rightarrow 0,n\rightarrow \infty}\frac{\log\left (\sum_{Y\in S_Y}e^{g_{tot}^{(1)}}\right )}{nn_r}\nonumber \\
& =  \lim_{n\rightarrow \infty} \frac{\beta^2(1-p_1)}{2n} \nonumber \\
& \quad +\lim_{n_r\rightarrow 0,n\rightarrow \infty}
\frac{\log\left (\sum_{Y\in S_Y}
\prod_{k=1}^{2}e^{\frac{\alpha_{3-k}^2\sum_{i=1}^{\frac{n_r}{m_k}}\|\1^TY^{(i,k)}\|_2^2}{2}}
\int_{\z^{(3)}}e^{\alpha_{3} \sum_{i=1}^{\frac{n_r}{m_1}}\1^TY^{(i,1)}\z^{(3)}}D\z^{(3)}\right )}{nn_r}\nonumber \\
& =  \lim_{n\rightarrow \infty}\frac{\beta^2(1-p_1)}{2n}\nonumber \\
& \quad
+\lim_{n_r\rightarrow 0,n\rightarrow \infty}\frac{\log\left (\int_{\z^{(3)}}\sum_{Y\in S_Y}
\prod_{k=1}^{2}e^{\frac{\alpha_{3-k}^2\sum_{i=1}^{\frac{n_r}{m_k}}\|\1^TY^{(i,k)}\|_2^2}{2}}
e^{\alpha_{3}\sum_{i=1}^{\frac{n_r}{m_1}}\1^TY^{(i,1)}\z^{(3)}}D\z^{(3)}\right )}{nn_r}\nonumber \\
& =  \lim_{n\rightarrow \infty}\frac{\beta^2(1-p_1)}{2n} \nonumber \\
& \quad +\lim_{n_r\rightarrow 0,n\rightarrow \infty}
\frac{\log\left (\int_{\z^{(3)}}\left (\sum_{Y^{(i,1)}\in S_Y^{(m_1)}}
e^{\frac{\alpha_{1}^2\sum_{i=1}^{\frac{m_1}{m_2}}\|\1^TY^{(i,2)}\|_2^2}
{2}}
e^{\frac{\alpha_{2}^2\|\1^TY^{(i,1)}\|_2^2}
{2}}
e^{\alpha_{3}\1^TY^{(i,1)}\z^{(3)}}\right )^{\frac{n_r}{m_1}}D\z^{(3)}\right )}{nn_r}, \nonumber \\\label{eq:2rsbsaddleI3gtotsqrt}
\end{align}
where $S_Y^{(m_1)}=\underbrace{\cY\times \cY \times \dots \times \cY}_{m_1\hspace{.03in} times}$ and $D\z^{(3)}=\frac{e^{-\frac{\|\z^{(3)}\|_2^2}{2}}d\z^{(3)}}{(2\pi)^{\frac{n}{2}}}$ A bit of further algebraic transformations gives
\begin{align}
\hspace{-.0in}I_{P}^{(2)} & =  \lim_{n\rightarrow \infty}\frac{\beta^2(1-p_1)}{2n} \nonumber \\
& \quad +\lim_{n_r\rightarrow 0,n\rightarrow \infty}
\frac{\log\left (\int_{\z^{(3)}}
e^{\log\left (\sum_{Y^{(i)}\in S_Y^{(m_1)}}
e^{\frac{\alpha_{1}^2\sum_{i=1}^{\frac{m_1}{m_2}}\|\1^TY^{(i,2)}\|_2^2}
{2}}
e^{\frac{\alpha_{2}^2\|\1^TY^{(i,1)}\|_2^2}
{2}}e^{\alpha_{3} \1^TY^{(i,1)}\z^{(3)}}\right )^{\frac{n_r}{m_1}}}D\z^{(3)}\right )}{nn_r}\nonumber \\
& = \lim_{n\rightarrow \infty}\frac{\beta^2(1-p_1)}{2n} \nonumber \\
& \quad +\lim_{n_r\rightarrow 0,n\rightarrow \infty}\frac{\log\left (\int_{\z^{(3)}}e^{\frac{n_r}{m_1}\log\left (\sum_{Y^{(i)}\in S_Y^{(m_1)}}
e^{\frac{\alpha_{1}^2\sum_{i=1}^{\frac{m_1}{m_2}}\|\1^TY^{(i,2)}\|_2^2}
{2}}
e^{\frac{\alpha_{2}^2\|\1^TY^{(i,1)}\|_2^2}
{2}}e^{\alpha_{3}\1^TY^{(i,1)}\z^{(3)}}\right )}D\z^{(3)}\right )}{nn_r}.\nonumber \\ \label{eq:2rsbsaddleI3gtot1sqrt}
\end{align}
Since $n_r\rightarrow 0$ we also have
{\small \begin{align}
\hspace{-.0in}I_{P}^{(2)}
\hspace{-.0in} & =  \lim_{n\rightarrow \infty} \frac{\beta^2(1-p_1)}{2n} \nonumber \\
& \quad +\lim_{n_r\rightarrow 0,n\rightarrow \infty}\frac{\log\left (1+\frac{n_r}{m_1}\int_{\z^{(3)}}\log\left (\sum_{Y^{(i)}\in S_Y^{(m_1)}}
e^{\frac{\alpha_{1}^2\sum_{i=1}^{\frac{m_1}{m_2}}\|\1^TY^{(i,2)}\|_2^2}
{2}}
e^{\frac{\alpha_{2}^2\|\1^TY^{(i,1)}\|_2^2}
{2}}e^{\alpha_{3}\1^TY^{(i,1)}\z^{(3)}}\right )D\z^{(3)}\right )}{nn_r} \nonumber  \\
 & = \lim_{n\rightarrow \infty} \frac{\beta^2(1-p_1)}{2n} \nonumber \\
& \quad +\lim_{n\rightarrow \infty}\frac{\frac{1}{m_1}\int_{\z^{(3)}}\log\left (\sum_{Y^{(i)}\in S_Y^{(m_1)}}
 e^{\frac{\alpha_{1}^2\sum_{i=1}^{\frac{m_1}{m_2}}\|\1^TY^{(i,2)}\|_2^2}
{2}}
 e^{\frac{\alpha_{2}^2\|\1^TY^{(i,1)}\|_2^2}
{2}}e^{\alpha_{3}\1^TY^{(i,1)}\z^{(3)}}\right )D\z^{(3)}}{n}. \nonumber \\
\label{eq:2rsbsaddleI3gtot2sqrt}
\end{align}}

\noindent Transforming further, we also have
\begin{align}
  I_{P}^{(2)}
   & = \lim_{n\rightarrow \infty} \frac{\beta^2(1-p_1)}{2n} \nonumber \\
& \quad +\lim_{n\rightarrow \infty}\frac{\int_{\z^{(3)}}\log\left (\sum_{Y^{(i)}\in S_Y^{(m_1)}}
 e^{\frac{\alpha_{1}^2\sum_{i=1}^{\frac{m_1}{m_2}}\|\1^TY^{(i,2)}\|_2^2}
{2}}
 e^{\frac{\alpha_{2}^2\|\1^TY^{(i,1)}\|_2^2}
{2}}e^{\alpha_{3}\1^TY^{(i,1)}\z^{(3)}}\right )D\z^{(3)}}{m_1n} \nonumber \\
\hspace{-.0in} & = \lim_{n\rightarrow \infty}\frac{\beta^2(1-p_1)}{2n} \nonumber \\
 & \quad +\lim_{n\rightarrow \infty}\frac{\int_{\z^{(3)}}\log\left (\sum_{Y^{(i)}\in S_Y^{(m_1)}}
  e^{\frac{\alpha_{1}^2\sum_{i=1}^{\frac{m_1}{m_2}}\|\1^TY^{(i,2)}\|_2^2}
{2}}
 \int_{\z^{(2)}}e^{\alpha_{2}\1^TY^{(i,1)}\z^{(2)}}
D\z^{(2)}
e^{\alpha_{3}\1^TY^{(i,1)}\z^{(3)}}\right )D\z^{(3)}}{m_1n} \nonumber \\
\hspace{-.0in} & = \lim_{n\rightarrow \infty}\frac{\beta^2(1-p_1)}{2n}  \nonumber \\
 & \quad +\lim_{n\rightarrow \infty}\frac{\int_{\z^{(3)}}\log\left (\int_{\z^{(2)}}\sum_{Y^{(i)}\in S_Y^{(m_1)}}
   e^{\frac{\alpha_{1}^2\sum_{i=1}^{\frac{m_1}{m_2}}\|\1^TY^{(i,2)}\|_2^2}
{2}}
 e^{\1^TY^{(i,1)}(\alpha_{2}\z^{(2)}+\alpha_{3}\z^{(3)})}
 D\z^{(2)}\right )D\z^{(3)}}{m_1n} \nonumber \\
\hspace{-.0in} & = \lim_{n\rightarrow \infty}\frac{\beta^2(1-p_1)}{2n}  \nonumber \\
 & \quad +\lim_{n\rightarrow \infty}\frac{\int_{\z^{(3)}}\log\left (\int_{\z^{(2)}}\sum_{Y^{(i)}\in S_Y^{(m_1)}}
   e^{\frac{\alpha_{1}^2\sum_{i=1}^{\frac{m_1}{m_2}}\|\1^TY^{(i,2)}\|_2^2}
{2}}
 e^{\sum_{i=1}^{\frac{m_1}{m_2}}\1^TY^{(i,2)}(\alpha_{2}\z^{(2)}+\alpha_{3}\z^{(3)})}
 D\z^{(2)}\right )D\z^{(3)}}{m_1n} \nonumber \\
\hspace{-.0in} & = \lim_{n\rightarrow \infty}\frac{\beta^2(1-p_1)}{2n}  \nonumber \\
 & \quad +\lim_{n\rightarrow \infty}\frac{\int_{\z^{(3)}}\log\left (\int_{\z^{(2)}}
 \lp \sum_{Y^{(i,2)}\in S_Y^{(m_2)}}
   e^{\frac{\alpha_{1}^2 \|\1^TY^{(i,2)}\|_2^2}
{2}}
 e^{\1^TY^{(i,2)}(\alpha_{2}\z^{(2)}+\alpha_{3}\z^{(3)})}\rp^{{\frac{m_1}{m_2}}}
 D\z^{(2)}\right )D\z^{(3)}}{m_1n} \nonumber \\
\hspace{-.0in} & = \lim_{n\rightarrow \infty}\frac{\beta^2(1-p_1)}{2n}  \nonumber \\
 & \quad +\lim_{n\rightarrow \infty}\frac{\int_{\z^{(3)}}\log\left (\int_{\z^{(2)}}
 \lp \sum_{Y^{(i,2)}\in S_Y^{(m_2)}}
\int_{\z^{(1)}}   e^{\alpha_{1} \1^TY^{(i,2)}\z^{(1)}}D\z^{(1)}
 e^{\1^TY^{(i,2)}(\alpha_{2}\z^{(2)}+\alpha_{3}\z^{(3)})}\rp^{{\frac{m_1}{m_2}}}
 D\z^{(2)}\right )D\z^{(3)}}{m_1n} \nonumber \\
\hspace{-.0in} & = \lim_{n\rightarrow \infty}\frac{\beta^2(1-p_1)}{2n}  \nonumber \\
 & \quad +\lim_{n\rightarrow \infty}\frac{\int_{\z^{(3)}}\log\left (\int_{\z^{(2)}}
 \lp \int_{\z^{(1)}}  \sum_{Y^{(i,2)}\in S_Y^{(m_2)}}
  e^{\alpha_{1} \1^TY^{(i,2)}\z^{(1)}}
 e^{\1^TY^{(i,2)}(\alpha_{2}\z^{(2)}+\alpha_{3}\z^{(3)})} D\z^{(1)} \rp^{{\frac{m_1}{m_2}}}
 D\z^{(2)}\right )D\z^{(3)}}{m_1n} \nonumber \\
\hspace{-.0in} & = \lim_{n\rightarrow \infty}\frac{\beta^2(1-p_1)}{2n}  \nonumber \\
 & \quad +\lim_{n\rightarrow \infty}\frac{\int_{\z^{(3)}}\log\left (\int_{\z^{(2)}}
 \lp \int_{\z^{(1)}}
 \lp \sum_{\y\in \cY}
 e^{\y^T(\alpha_{1}\z^{(1)}+\alpha_{2}\z^{(2)}+\alpha_{3}\z^{(3)})} \rp^{m_2}
 D\z^{(1)} \rp^{{\frac{m_1}{m_2}}}
 D\z^{(2)}\right )D\z^{(3)}}{m_1n}, \nonumber \\
  \label{eq:2rsbsaddleI3gtot3sqrt}
\end{align}
where $S_Y^{(m_2)}=\underbrace{\cY\times \cY \times \dots \times \cY}_{m_2\hspace{.03in} times}$, $D\z^{(1)}=\frac{e^{-\frac{\|\z^{(1)}\|_2^2}{2}}d\z^{(1)}}{(2\pi)^{\frac{n}{2}}}$, and $D\z^{(2)}=\frac{e^{-\frac{\|\z^{(2)}\|_2^2}{2}}d\z^{(2)}}{(2\pi)^{\frac{n}{2}}}$. It is not that difficult to see that $I_{P}^{(2)}$ and (\ref{eq:2rsbsaddleI3gtot3sqrt}) can be rewritten more compactly as
\begin{equation}
  I_{P}^{(2)}= \lim_{n\rightarrow \infty}\frac{E_{\z^{(3)}}\log\left (E_{\z^{(2)}}\left ( E_{\z^{(1)}}\left ( E_{\z^{(0)}}\left (\sum_{\y\in\cY}e^{\y^T ( \alpha_{0}\z^{(0)}+\alpha_{1}\z^{(1)}+\alpha_{2}\z^{(2)}+\alpha_{3} \z^{(3)})}
\right )\right )^{m_2}\right )^{\frac{m_1}{m_2}} \rp}{m_1n},\label{eq:2rsbfinalsaddleI3sqrtcomp}
\end{equation}
where
\begin{eqnarray}
\alpha_{0} & = & \beta\sqrt{1-p_2} \nonumber \\
\alpha_{1} & = & \beta\sqrt{p_2-p_1}\nonumber \\
\alpha_{2} & = & \beta\sqrt{p_1-p_0}\nonumber \\
\alpha_{3} & = & \beta\sqrt{p_0}.\label{eq:2rsbalpha1234sqrt}
\end{eqnarray}

\noindent \emph{\underline{3) Determining $I_{Q}^{(2)}$}}
\vspace{.1in}

\noindent As $I_{Q}^{(2)}$ and $I_{P}^{(2)}$ are structurally identical, one can repeat all the above steps and obtain
completely analogously to (\ref{eq:2rsbfinalsaddleI2sqrtcomp})
\begin{equation}
  I_{Q}^{(2)}= \lim_{n\rightarrow \infty}\frac{E_{\z^{(3)}}\log\left (E_{\z^{(2)}}\left ( E_{\z^{(1)}}\left ( E_{\z^{(0)}}\left (\sum_{\x\in\cX}e^{\c^T ( \bar{\alpha}_{0}\z^{(0)}+\bar{\alpha}_{1}\z^{(1)}+\bar{\alpha}_{2}\z^{(2)}+\alpha_{3} \z^{(3)})}
\right )\right )^{m_2}\right )^{\frac{m_1}{m_2}} \rp}{m_1n},\label{eq:2rsbfinalsaddleI2sqrtcomp}
\end{equation}
\begin{eqnarray}
\bar{\alpha}_{0} & = & \beta\sqrt{1-\bq_2} \nonumber \\
\bar{\alpha}_{1} & = & \beta\sqrt{\bq_2-\bq_1}\nonumber \\
\bar{\alpha}_{2} & = & \beta\sqrt{\bq_1-\bq_0}\nonumber \\
\bar{\alpha}_{3} & = & \beta\sqrt{\bq_0}.\label{eq:2rsbalphabar1234sqrt}
\end{eqnarray}

A combination of (\ref{eq:ElogZ1sqrt}), (\ref{eq:2rsbdefssaddleIs2rsbsqrt}), (\ref{eq:2rsbfinalsaddleI12rsbsqrt}), (\ref{eq:2rsbfinalsaddleI3sqrtcomp}), and (\ref{eq:2rsbfinalsaddleI2sqrtcomp}), gives $\bar{f}_{sq}^{(2)}(\beta)$ as a $2$-rsb estimate of $f_{sq}(\beta)$,
\begin{eqnarray}
\bar{f}_{sq}^{(2)}(\beta) &  = & -\lim_{n\rightarrow \infty} \frac{\beta^2}{2\frac{\beta}{\sqrt{n}}n}
(1-\bq_2p_2 +m_2(\bq_2p_2-\bq_1p_1)+m_1(\bq_1p_1-\bq_0p_0) ) \nonumber \\
& & +\lim_{n\rightarrow \infty}\frac{E_{\z^{(3)}}\log\left (E_{\z^{(2)}}\left ( E_{\z^{(1)}}\left ( E_{\z^{(0)}}\left (\sum_{\y\in\cY}e^{\y^T ( \alpha_{0}\z^{(0)}+\alpha_{1}\z^{(1)}+\alpha_{2}\z^{(2)}+\alpha_{3} \z^{(3)})}
\right )\right )^{m_2}\right )^{\frac{m_1}{m_2}} \rp}{\frac{\beta}{\sqrt{n}}m_1n} \nonumber \\
& & +\lim_{n\rightarrow \infty}\frac{E_{\z^{(3)}}\log\left (E_{\z^{(2)}}\left ( E_{\z^{(1)}}\left ( E_{\z^{(0)}}\left (\sum_{\x\in\cX}e^{\c^T ( \bar{\alpha}_{0}\z^{(0)}+\bar{\alpha}_{1}\z^{(1)}+\bar{\alpha}_{2}\z^{(2)}+\alpha_{3} \z^{(3)})}
\right )\right )^{m_2}\right )^{\frac{m_1}{m_2}} \rp}{\frac{\beta}{\sqrt{n}}  m_1n}, \nonumber \\\label{eq:finalfreeenergycompsqrt}
\end{eqnarray}
or in a fully compact form
\begin{eqnarray}
\bar{f}_{sq}^{(2)}(\beta) & = &
-\lim_{n\rightarrow \infty}
\frac{E_{z_3}\log \left ( E_{z_2}\left ( E_{z_1}\left (E_{z_0}e^{(\tilde{\alpha}_{0}  z_0+\tilde{\alpha}_{1}  z_1+\tilde{\alpha}_{2}  z_2+\tilde{\alpha}_{3}  z_3)}\right )^{m_2}\right )^{\frac{m_1}{m_2}}\rp }{\frac{\beta}{\sqrt{n}}m_1n}
\\
& & +\lim_{n\rightarrow \infty}\frac{E_{\z^{(3)}}\log\left (E_{\z^{(2)}}\left ( E_{\z^{(1)}}\left ( E_{\z^{(0)}}\left (\sum_{\y\in\cY}e^{\y^T ( \alpha_{0}\z^{(0)}+\alpha_{1}\z^{(1)}+\alpha_{2}\z^{(2)}+\alpha_{3} \z^{(3)})}
\right )\right )^{m_2}\right )^{\frac{m_1}{m_2}} \rp}{\frac{\beta}{\sqrt{n}}m_1n} \nonumber \\
& & +\lim_{n\rightarrow \infty}\frac{E_{\z^{(3)}}\log\left (E_{\z^{(2)}}\left ( E_{\z^{(1)}}\left ( E_{\z^{(0)}}\left (\sum_{\x\in\cX}e^{\c^T ( \bar{\alpha}_{0}\z^{(0)}+\bar{\alpha}_{1}\z^{(1)}+\bar{\alpha}_{2}\z^{(2)}+\alpha_{3} \z^{(3)})}
\right )\right )^{m_2}\right )^{\frac{m_1}{m_2}} \rp}{\frac{\beta}{\sqrt{n}}  m_1n}, \nonumber \\
\label{eq:2rsbfinalfreeenergycomp1sqrt}
\end{eqnarray}
where
\begin{eqnarray}
\tilde{\alpha}_{0} & = & \beta\sqrt{1-\bq_2p_2}\nonumber \\
\tilde{\alpha}_{2} & = & \beta\sqrt{\bq_2p_2-\bq_1p_1}\nonumber \\
\tilde{\alpha}_{2} & = & \beta\sqrt{\bq_1p_1-\bq_0p_0}\nonumber \\
\tilde{\alpha}_{3} & = & \beta\sqrt{\bq_0p_0}.\label{eq:tildealpha1234sqrt}
\end{eqnarray}

The following then summarizes what we presented in this subsection.
\begin{observation} (\textbf{2-rsb})
Let the thermodynamic limit of the averaged free energy of generic bilinear Hamiltonian based model, $f_{sq}(\beta)$, be as defined in (\ref{eq:logpartfunsqrt}). Let matrices $P^{(2)}$
and $Q^{(2)}$ be as defined in (\ref{eq:2rsbQsaddle2rsbsqrt}). Further let $\alpha$, $\bar{\alpha}$, and $\tilde{\alpha}$ sequences be as defined in
(\ref{eq:2rsbalpha1234sqrt}), (\ref{eq:2rsbalphabar1234sqrt}), and (\ref{eq:tildealpha1234sqrt})
Then, for appropriately selected parameters in the definition of $P^{(2)}$ and $Q^{(2)}$, one has that the  $2$-rsb estimate for $f_{sq}(\beta)$, $\bar{f}_{sq}^{(2)}(\beta)$, is as given in (\ref{eq:2rsbfinalfreeenergycomp1sqrt}).
\label{obs:2rsbsqrt}
\end{observation}
Observation \ref{obs:2rsbsqrt} is sufficient to analyze the average free energy and its  ground state on the second level of symmetry breaking. The first two levels contain all the necessary ingredients to formalize general mechanism for $r$ levels of symmetry breaking where $r\in\mN$. In the following subsection, we present the final results of such a formalization.

\subsubsection{An $r$-rsb scheme based on Parisi ansatz}
\label{sec:rrsbsqrtposhop}

Setting $p_{-1}=\bq_{-1}=0$ and $m_0=n_r$, within the Parisi ansatz, $P$ and $Q$ have the following structure on the $r$-th level of the symmetry breaking
\begin{eqnarray}
P^{(r)}=(1-p_r)I_{n_r}
+\sum_{k=0}^rI_{\frac{n_r}{m_k}}\otimes (p_k-p_{k-1})U_{m_{k}\times m_{k}} \nonumber \\
Q^{(r)}=(1-\bq_r)I_{n_r}
+\sum_{k=0}^rI_{\frac{n_r}{m_k}}\otimes (\bq_k-\bq_{k-1})U_{m_{k}\times m_{k}},
\label{eq:rrsbQsaddle2rsbsqrt}
\end{eqnarray}
where $I$  and $U$, are as in Section \ref{sec:1rsbsqrtposhop}. One then has
\begin{eqnarray}
n_r-\frac{n_r}{m_r} \quad \mbox{eigenvalues equal to}: & & \lambda_{1}^{(r)}=1-p_r\nonumber \\
\frac{n_r}{m_r}-\frac{n_r}{m_{r-1}} \quad \mbox{eigenvalues equal to}: & & \lambda_{2}^{(r)}=(1-p_r)+m_r(p_r-p_{r-1})\nonumber \\
\vdots  & & \vdots \nonumber \\
\frac{n_r}{m_2}-\frac{n_r}{m_1} \quad \mbox{eigenvalues equal to}: & & \lambda_{r}^{(r)}=(1-p_r)+ \sum_{k=2}^{r}m_k(p_k-p_{k-1})\nonumber \\
\frac{n_r}{m_1}-1 \quad \mbox{eigenvalues equal to}: & & \lambda_{r+1}^{(r)}=(1-p_r)+ \sum_{k=1}^{r}m_k(p_k-p_{k-1})\nonumber \\
1 \quad \mbox{eigenvalue equal to}: & & \lambda_{r+2}^{(r)}=(1-p_r)+ \sum_{k=1}^{r}m_k(p_k-p_{k-1})+n_r\beta p_0,\nonumber \\\label{eq:rrsbeigPsaddle2rsbsqrt}
\end{eqnarray}
for the eigenvalues of matrix $P^{(r)}$ and
\begin{eqnarray}
n_r-\frac{n_r}{m_r} \quad \mbox{eigenvalues equal to}: & & \lambda_{1}^{(r)}=1-\bq_r\nonumber \\
\frac{n_r}{m_r}-\frac{n_r}{m_{r-1}} \quad \mbox{eigenvalues equal to}: & & \lambda_{2}^{(r)}=(1-\bq_r)+m_r(\bq_r-\bq_{r-1})\nonumber \\
\vdots  & & \vdots \nonumber \\
\frac{n_r}{m_2}-\frac{n_r}{m_1} \quad \mbox{eigenvalues equal to}: & & \lambda_{r}^{(r)}=(1-\bq_r)+ \sum_{k=2}^{r}m_k(\bq_k-\bq_{k-1})\nonumber \\
\frac{n_r}{m_1}-1 \quad \mbox{eigenvalues equal to}: & & \lambda_{r+1}^{(r)}=(1-\bq_r)+ \sum_{k=1}^{r}m_k(\bq_k-\bq_{k-1})\nonumber \\
1 \quad \mbox{eigenvalue equal to}: & & \lambda_{r+2}^{(r)}=(1-\bq_r)+ \sum_{k=1}^{r}m_k(\bq_k-\bq_{k-1})+n_r\beta \bq_0,\nonumber \\\label{eq:rrsbeigQsaddle2rsbsqrt}
\end{eqnarray}
for the eigenvalues of matrix $Q^{(r)}$. Analogously to what was done in Section \ref{sec:1rsbsqrtposhop}, the computation of $s_{sq}(\beta,P^{(r)},Q^{(r)})$ is again split into three steps, after one recognizes
\begin{eqnarray}
\lim_{n_r\rightarrow 0,n\rightarrow \infty}  \frac{s_{sq}(\beta,P^{(r)},Q^{(r)})}{n_r}
& = &  \lim_{n_r\rightarrow 0,n\rightarrow \infty}\frac{1}{n_r}
\Bigg( \Bigg.   -\frac{\beta^2}{2n}\tr(Q^{(r)}P^{(r)})
+  \log\left (\sum_{X \in S_X}e^{\frac{\beta^2}{2}\tr(Q^{(r)} X^TX)}\right )^{\frac{1}{n}} \nonumber \\
& & +  \log\left (\sum_{Y \in S_Y}e^{\frac{\beta^2}{2}\tr(P^{(r)} Y^TY)}\right )^{\frac{1}{n}}  \Bigg.\Bigg) \nonumber \\
& = &  I_{Q,P}^{(r)}+I_{Q}^{(r)}+I_{P}^{(r)}, \label{eq:rrsbdefssaddleIs2rsbsqrt}
\end{eqnarray}
where
\begin{eqnarray}
 I_{Q,P}^{(r)} & = & \lim_{n_r\rightarrow 0,n\rightarrow \infty}\frac{-\frac{\beta^2}{2n}\tr(Q^{(r)}P^{(r)})}{n_r}\nonumber \\
I_{Q}^{(r)}& = & \lim_{n_r\rightarrow 0,n\rightarrow \infty}\frac{\log\left (\sum_{X \in S_X}e^{\frac{\beta^2}{2}\tr(Q^{(r)} X^TX)}\right )^{\frac{1}{n}}}{n_r}\nonumber \\
 I_{P}^{(r)}& = & \lim_{n_r\rightarrow 0,n\rightarrow \infty}\frac{\log\left (\sum_{Y \in S_Y}e^{\frac{\beta^2}{2}\tr(P^{(r)} Y^TY)}\right )^{\frac{1}{n}}}{n_r}. \label{eq:rrsbsaddleIs2rsbsqrt}
\end{eqnarray}

\noindent \emph{\underline{1) Determining $I_{Q,P}^{(r)}$}}
\vspace{.1in}

\noindent Following what was done in Sections \ref{sec:1rsbsqrtposhop} and \ref{sec:2rsbsqrtposhop}, we now obtain
\begin{align}
\hspace{-.0in}I_{Q,P}^{(r)}
& = \lim_{n\rightarrow \infty}-\frac{\beta^2}{2n}\lp (1-\bq_rp_r)+ \sum_{k=1}^{r}m_k(\bq_kp_k-\bq_{k-1}p_{k-1}) \rp.\label{eq:rrsbfinalsaddleI12rsbsqrt}
\end{align}

\noindent \emph{\underline{2) Determining $I_{P}^{(r)}$}}
\vspace{.1in}

\noindent  One starts with
\begin{equation}
 I_{P}^{(r)}= \lim_{n_r\rightarrow 0,n\rightarrow \infty}\frac{\log\left (\sum_{Y \in S_Y}e^{\frac{\beta^2}{2}\tr(P^{(r)} Y^TY)}\right )^{\frac{1}{n}}}{n_r}, \label{eq:2rsbsaddleI3sqrt}
\end{equation}
and introduces
\begin{equation}
\hspace{-0in}g_{tot}^{(r,P)}=\frac{\beta^2}{2}\tr(P^{(r)}Y^TY)=\frac{n n_r\beta^2(1-p_r)}{2}
+\sum_{k=1}^{r}\frac{\alpha_{r+1-k}^2\sum_{i_2=1}^{\frac{n_r}{m_k}}\|\1^TY^{(i,k)}\|_2^2}
{2}
 +\frac{\alpha_{3}^2 \|\1^TY^T\|_2^2}{2},\label{eq:2rsbg3tot1sqrt}
\end{equation}
  where
\begin{eqnarray}
Y^{(i,k)}=Y^T_{(i-1)m_k+1:im_k,1:n},1\leq i\leq \frac{n_r}{m_k}, k\in\{1,2,\dots,r\}.\label{eq:2rsbdefYisqrt}
\end{eqnarray}
and
\begin{eqnarray}
\alpha_{r+1-k} & = & \beta\sqrt{p_k-p_{k-1}}, k\in\{1,2,\dots,r\}.\label{eq:2rsbalphabar12sqrt}
\end{eqnarray}
Repeating line by line the strategy presented between (\ref{eq:2rsbsaddleI3gtotsqrt})-(\ref{eq:2rsbalpha1234sqrt}), one arrives at the following $r$-rsb analogue to (\ref{eq:2rsbfinalsaddleI3sqrtcomp})
\begin{equation}
  I_{P}^{(r)}= \lim_{n\rightarrow \infty}\frac{E_{\z_{P}^{(r+1)}}\log \lp \zeta_{r,P}\rp}{m_1n},\label{eq:rrsbfinalsaddleI3sqrtcomp}
\end{equation}
where
\begin{eqnarray}
\label{eq:rrsbp}
\zeta_{k,P} & = & \mE_{\z_P^{(k)}}\zeta_{k-1,P}^{\frac{m_{r+1-k}}{m_{r+2-k}}},k\in\{1,2,\dots,r\} \nonumber \\
\zeta_{-1,P} & = & Z_{\y,P} \nonumber \\
Z_{\y,P} & = & \sum_{\y\in\cY}e^{\y^T \sum_{k=0}^{r+1} \alpha_{k}\z_{P}^{(k)}}\nonumber \\
\alpha_{r+1-k} & = & \beta\sqrt{p_k-p_{k-1}}, k\in\{0,1,2,\dots,r+1\},
 \end{eqnarray}
and $\z_{P}^{(k)}$ are $m$-dimensional vectors of iid standard normals, and $m_{r+1}=m_{r+2}=1$, $p_{r+1}=1$, and $p_{-1}=0$.

\noindent \emph{\underline{3) Determining $I_{Q}^{(r)}$}}
\vspace{.1in}

\noindent As $I_{Q}^{(r)}$ and $I_{P}^{(r)}$ are structurally identical, we immediately have
\begin{equation}
  I_{Q}^{(r)}= \lim_{n\rightarrow \infty}\frac{E_{\z_{Q}^{(r+1)}}\log \lp \zeta_{r,Q}\rp}{m_1n},\label{eq:rrsbfinalsaddleI2sqrtcomp}
\end{equation}
where
\begin{eqnarray}
\label{eq:rrsbq}
\zeta_{k,Q} & = & \mE_{\z_Q^{(k)}}\zeta_{k-1,Q}^{\frac{m_{r+1-k}}{m_{r+2-k}}},k\in\{1,2,\dots,r\} \nonumber \\
\zeta_{-1,Q} & = & Z_{\x,Q} \nonumber \\
Z_{\x,Q} & = & \sum_{\x\in\cX}e^{\x^T \sum_{k=0}^{r+1} \alpha_{k}\z_{Q}^{(k)}}\nonumber \\
\bar{\alpha}_{r+1-k} & = & \beta\sqrt{\bq_k-\bq_{k-1}}, k\in\{0,1,2,\dots,r+1\},
 \end{eqnarray}
and $\z_{Q}^{(k)}$ are $n$-dimensional vectors of iid standard normals, and  $\bq_{r+1}=1$, and $\bq_{-1}=0$.

A combination of (\ref{eq:ElogZ1sqrt}), (\ref{eq:rrsbdefssaddleIs2rsbsqrt}), (\ref{eq:rrsbfinalsaddleI12rsbsqrt}), (\ref{eq:rrsbfinalsaddleI3sqrtcomp}), and (\ref{eq:rrsbfinalsaddleI2sqrtcomp}), gives $\bar{f}_{sq}^{(r)}(\beta)$ as an $r$-rsb estimate of $f_{sq}(\beta)$,
\begin{eqnarray}
\bar{f}_{sq}^{(r)}(\beta) &  = & -\lim_{n\rightarrow \infty} \frac{\beta^2}{2\frac{\beta}{\sqrt{n}}n}
\lp (1-\bq_rp_r)+ \sum_{k=1}^{r}m_k(\bq_kp_k-\bq_{k-1}p_{k-1}) \rp \nonumber \\
& & +\lim_{n\rightarrow \infty}\frac{E_{\z_P^{(r+1)}}\log\left (E_{\z_P^{(r)}}\dots\left ( E_{\z_P^{(1)}}\left ( E_{\z_P^{(0)}}\left (\sum_{\y\in\cY}e^{\y^T \sum_{k=0}^{r+1} \alpha_{k}\z_P^{(k)} }
\right )\right )^{m_r}\dots \right )^{\frac{m_1}{m_2}} \rp}{\frac{\beta}{\sqrt{n}}m_1n} \nonumber \\
& & +\lim_{n\rightarrow \infty}\frac{E_{\z_Q^{(r+1)}}\log\left (E_{\z_Q^{(r)}}\dots\left ( E_{\z_Q^{(1)}}\left ( E_{\z_Q^{(0)}}\left (\sum_{\x\in\cX}e^{\x^T \sum_{k=0}^{r+1} \alpha_{k}\z_Q^{(k)} }
\right )\right )^{m_r}\dots \right )^{\frac{m_1}{m_2}} \rp}{\frac{\beta}{\sqrt{n}}m_1n} \nonumber \\
, \nonumber \\ \label{eq:rrsbfinalfreeenergycompsqrt}
\end{eqnarray}
or in a fully compact form
\begin{eqnarray}
\bar{f}_{sq}^{(r)}(\beta) & = &
-\lim_{n\rightarrow \infty}\frac{E_{z_{P,Q}^{(r+1)}}\log \lp \zeta_{r,P,Q}\rp}{\frac{\beta}{\sqrt{n}}  m_1n}
 +\lim_{n\rightarrow \infty}\frac{E_{\z_{P}^{(r+1)}}\log \lp \zeta_{r,P}\rp}{\frac{\beta}{\sqrt{n}}  m_1n}
  +\lim_{n\rightarrow \infty}\frac{E_{\z_{Q}^{(r+1)}}\log \lp \zeta_{r,Q}\rp}{\frac{\beta}{\sqrt{n}}  m_1n}, \nonumber \\
\label{eq:rrsbfinalfreeenergycomp1sqrt}
\end{eqnarray}
where
\begin{eqnarray}
\label{eq:rrsbpq3}
\zeta_{k,P,Q} & = & \mE_{\z_Q^{(k)}}\zeta_{k-1,P,Q}^{\frac{m_{r+1-k}}{m_{r+2-k}}},k\in\{1,2,\dots,r\} \nonumber \\
\zeta_{-1,P,Q} & = & Z_{P,Q} \nonumber \\
Z_{P,Q} & = & e^{\sum_{k=0}^{r+1} \tilde{\alpha}_{k}z_{P,Q}^{(k)}}\nonumber \\
\tilde{\alpha}_{r+1-k} & = & \beta\sqrt{\bq_kp_k-\bq_{k-1}p_{k-1}}, k\in\{0,1,2,\dots,r+1\},
 \end{eqnarray}
and $z_{P,Q}^{(k)}$ are iid standard normals.

The following summarizes all of the above.
\begin{observation} (\textbf{r-rsb})
Let the thermodynamic limit of the averaged free energy of generic bilinear Hamiltonian based models, $f_{sq}(\beta)$, be as defined in (\ref{eq:logpartfunsqrt}). For any $r\in\mN$, let matrices $P^{(r)}$
and $Q^{(r)}$ be as defined in (\ref{eq:rrsbQsaddle2rsbsqrt}). Further let $\alpha$, $\bar{\alpha}$, and $\tilde{\alpha}$ sequences be as defined in
(\ref{eq:rrsbp}), (\ref{eq:rrsbq}), and (\ref{eq:rrsbpq3})
Then, for appropriately selected parameters in the definition of $P^{(r)}$ and $Q^{(r)}$, one has that the  $r$-rsb estimate for $f_{sq}(\beta)$, $\bar{f}_{sq}^{(r)}(\beta)$, is as given in (\ref{eq:rrsbfinalfreeenergycomp1sqrt}).
\label{obs:rrsbsqrt}
\end{observation}
Observation \ref{obs:rrsbsqrt} effectively establishes all that is needed to practically utilize the Parisi ansatz type of replica symmetry breaking.

\section{Fully lifted interpolation of bilinearly indexed random processes}
\label{sec:flrdt}

In \cite{Stojnicnflgscompyx23,Stojnicsflgscompyx23}  powerful tools for studying comparative behavior of bilinearly indexed random processes are introduced. They are based on the so-called fully lifted (fl) interpolation. In particular, the following theorem is among the key results of \cite{Stojnicsflgscompyx23}.

\begin{theorem}(\cite{Stojnicsflgscompyx23})
\label{thm:thm3}
Let $r\in\mN$ and $k\in\{1,2,\dots,r+1\}$ and consider vectors $\m=[\m_0,\m_1,\m_2,...,\m_r,\m_{r+1}]$ with $\m_0=1$ and $\m_{r+1}=0$,
$\p=[\p_0,\p_1,...,\p_r,\p_{r+1}]$ with $1\geq\p_0\geq \p_1\geq \p_2\geq \dots \geq \p_r\geq \p_{r+1} =0$ and $\q=[\q_0,\q_1,\q_2,\dots,\q_r,\q_{r+1}]$ with $1\geq\q_0\geq \q_1\geq \q_2\geq \dots \geq \q_r\geq \q_{r+1} = 0$. Let the components of $G\in\mR^{m\times n}$, $u^{(4,k)}\in\mR$, $\u^{(2,k)}\in\mR^m$, and $\h^{(k)}\in\mR^n$ be i.i.d. standard normals. Also, let $a_k=\sqrt{\p_{k-1}\q_{k-1}-\p_k\q_k}$, $b_k=\sqrt{\p_{k-1}-\p_{k}}$, and $c_k=\sqrt{\q_{k-1}-\q_{k}}$, and let ${\mathcal U}_k\triangleq [u^{(4,k)},\u^{(2,k)},\h^{(2k)}]$.  Assuming that set ${\mathcal X}=\{\x^{(1)},\x^{(2)},\dots,\x^{(l)}\}$, where $\x^{(i)}\in \mR^{n},1\leq i\leq l$, set ${\mathcal Y}=\{\y^{(1)},\y^{(2)},\dots,\y^{(l)}\}$, where $\y^{(i)}\in \mR^{m},1\leq i\leq l$, and scalars $\beta\geq 0$ and $s\in\mR$ are given, and consider the following function
\begin{equation}\label{eq:thm3eq1}
\psi(\calX,\calY,\p,\q,\m,\beta,s,t)  =  \mE_{G,{\mathcal U}_{r+1}} \frac{1}{\beta|s|\sqrt{n}\m_r} \log
\lp \mE_{{\mathcal U}_{r}} \lp \dots \lp \mE_{{\mathcal U}_2}\lp\lp\mE_{{\mathcal U}_1} \lp Z^{\m_1}\rp\rp^{\frac{\m_2}{\m_1}}\rp\rp^{\frac{\m_3}{\m_2}} \dots \rp^{\frac{\m_{r}}{\m_{r-1}}}\rp,
\end{equation}
where
\begin{eqnarray}\label{eq:thm3eq2}
Z & \triangleq & \sum_{i_1=1}^{l}\lp\sum_{i_2=1}^{l}e^{\beta D_0^{(i_1,i_2)}} \rp^{s} \nonumber \\
 D_0^{(i_1,i_2)} & \triangleq & \sqrt{t}(\y^{(i_2)})^T
 G\x^{(i_1)}+\sqrt{1-t}\|\x^{(i_2)}\|_2 (\y^{(i_2)})^T\lp\sum_{k=1}^{r+1}b_k\u^{(2,k)}\rp\nonumber \\
 & & +\sqrt{t}\|\x^{(i_1)}\|_2\|\y^{(i_2)}\|_2\lp\sum_{k=1}^{r+1}a_ku^{(4,k)}\rp +\sqrt{1-t}\|\y^{(i_2)}\|_2\lp\sum_{k=1}^{r+1}c_k\h^{(k)}\rp^T\x^{(i)}
 \end{eqnarray}
 Let
  \begin{eqnarray}\label{eq:rthlev2genanal7a}
\zeta_r\triangleq \mE_{{\mathcal U}_{r}} \lp \dots \lp \mE_{{\mathcal U}_2}\lp\lp\mE_{{\mathcal U}_1} \lp Z^{\frac{\m_1}{\m_0}}\rp\rp^{\frac{\m_2}{\m_1}}\rp\rp^{\frac{\m_3}{\m_2}} \dots \rp^{\frac{\m_{r}}{\m_{r-1}}}, r\geq 1.
\end{eqnarray}
One then has
\begin{eqnarray}\label{eq:rthlev2genanal7b}
\zeta_k = \mE_{{\mathcal U}_{k}} \lp  \zeta_{k-1} \rp^{\frac{\m_{k}}{\m_{k-1}}}, k\geq 2,\quad \mbox{and} \quad
\zeta_1=\mE_{{\mathcal U}_1} \lp Z^{\frac{\m_1}{\m_0}}\rp,
\end{eqnarray}
with $\zeta_0=Z$ set for the completeness.  Moreover, consider the operators
\begin{eqnarray}\label{eq:thm3eq3}
 \Phi_{{\mathcal U}_k} & \triangleq &  \mE_{{\mathcal U}_{k}} \frac{\zeta_{k-1}^{\frac{\m_k}{\m_{k-1}}}}{\zeta_k},
 \end{eqnarray}
and set
\begin{eqnarray}\label{eq:thm3eq4}
  \gamma_0(i_1,i_2) & = &
\frac{(C^{(i_1)})^{s}}{Z}  \frac{A^{(i_1,i_2)}}{C^{(i_1)}} \nonumber \\
\gamma_{01}^{(r)}  & = & \prod_{k=r}^{1}\Phi_{{\mathcal U}_k} (\gamma_0(i_1,i_2)) \nonumber \\
\gamma_{02}^{(r)}  & = & \prod_{k=r}^{1}\Phi_{{\mathcal U}_k} (\gamma_0(i_1,i_2)\times \gamma_0(i_1,p_2)) \nonumber \\
\gamma_{k_1+1}^{(r)}  & = & \prod_{k=r}^{k_1+1}\Phi_{{\mathcal U}_k} \lp \prod_{k=k_1}^{1}\Phi_{{\mathcal U}_k}\gamma_0(i_1,i_2)\times \prod_{k=k_1}^{1} \Phi_{{\mathcal U}_k}\gamma_0(p_1,p_2) \rp.
 \end{eqnarray}
Also, let
\begin{eqnarray}\label{eq:thm3eq41}
\phi^{(k_1,\p)} & = &
  -(1-t)\lp \m_{k_1}-\m_{k_1+1}\rp \mE_{G,{\mathcal U}_{r+1}} \langle \|\x^{(i_1)}\|_2\|\x^{(p_1)}\|_2(\y^{(p_2)})^T\y^{(i_2)} \rangle_{\gamma_{k_1+1}^{(1)}} \nonumber \\
& &   - t \q_{k_1}
\lp \m_{k_1}-\m_{k_1+1}\rp  \mE_{G,{\mathcal U}_{r+1}} \langle\|\x^{(i_1)}\|_2\|\x^{(p_1)}\|_2\|\y^{(i_2)}\|_2\|\y^{(p_2)}\|_2\rangle_{\gamma_{k_1+1}^{(1)}}\nonumber \\
\phi^{(k_1,\q)} & = &
   -(1-t)  \lp \m_{k_1}-\m_{k_1+1}\rp  \mE_{G,{\mathcal U}_{r+1}} \langle \|\y^{(i_2)}\|_2\|\y^{(p_2)}\|_2(\x^{(i_1)})^T\x^{(p_1)}\rangle_{\gamma_{k_1+1}^{(1)}} \nonumber \\
& &   - t \p_{k_1}
\lp \m_{k_1}-\m_{k_1+1}\rp  \mE_{G,{\mathcal U}_{r+1}} \langle\|\x^{(i_1)}\|_2\|\x^{(p_1)}\|_2\|\y^{(i_2)}\|_2\|\y^{(p_2)}\|_2\rangle_{\gamma_{k_1+1}^{(1)}}.
\end{eqnarray}
Then for $k_1\in\{1,2,\dots,r\}$
\begin{eqnarray}\label{eq:thm3eq42}
\frac{d\psi(\calX,\calY,\p,\q,\m,\beta,s,t)}{d\p_{k_1}}  & = &       \frac{\mbox{sign}(s)s\beta}{2\sqrt{n}} \phi^{(k_1,\p)}\nonumber \\
\frac{d\psi(\calX,\calY,\p,\q,\m,\beta,s,t)}{d\q_{k_1}}  & = &       \frac{\mbox{sign}(s)s\beta}{2\sqrt{n}} \phi^{(k_1,\q)}.
 \end{eqnarray}
 \end{theorem}
The following fundamental result from \cite{Stojnicnflgscompyx23} is another essential ingredient.
\begin{theorem}(\cite{Stojnicnflgscompyx23})
\label{thm:thm4}
Assume the setup of Theorem \ref{thm:thm3}. Let also
\begin{eqnarray}\label{eq:thm3eq5}
 \phi_{k_1}^{(r)} & = &
-s(\m_{k_1-1}-\m_{k_1}) \nonumber \\
&  & \times
\mE_{G,{\mathcal U}_{r+1}} \langle (\p_{k_1-1}\|\x^{(i_1)}\|_2\|\x^{(p_1)}\|_2 -(\x^{(p_1)})^T\x^{(i_1)})(\q_{k_1-1}\|\y^{(i_2)}\|_2\|\y^{(p_2)}\|_2 -(\y^{(p_2)})^T\y^{(i_2)})\rangle_{\gamma_{k_1}^{(r)}} \nonumber \\
 \phi_{01}^{(r)} & = & (1-\p_0)(1-\q_0)\mE_{G,{\mathcal U}_{r+1}}\langle \|\x^{(i_1)}\|_2^2\|\y^{(i_2)}\|_2^2\rangle_{\gamma_{01}^{(r)}} \nonumber\\
\phi_{02}^{(r)} & = & (s-1)(1-\p_0)\mE_{G,{\mathcal U}_{r+1}}\left\langle \|\x^{(i_1)}\|_2^2 \lp\q_0\|\y^{(i_2)}\|_2\|\y^{(p_2)}\|_2-(\y^{(p_2)})^T\y^{(i_2)}\rp\right\rangle_{\gamma_{02}^{(r)}}. \end{eqnarray}
Then
\begin{eqnarray}\label{eq:thm3eq6}
\frac{d\psi(\calX,\calY,\p,\q,\m,\beta,s,t)}{dt}  & = &       \frac{\mbox{sign}(s)\beta}{2\sqrt{n}} \lp  \lp\sum_{k_1=1}^{r+1} \phi_{k_1}^{(r)}\rp +\phi_{01}^{(r)}+\phi_{02}^{(r)}\rp.
 \end{eqnarray}
It particular, choosing $\p_0=\q_0=1$, one also has
\begin{eqnarray}\label{eq:rthlev2genanal43}
\frac{d\psi(\calX,\calY,\p,\q,\m,\beta,s,t)}{dt}  & = &       \frac{\mbox{sign}(s)\beta}{2\sqrt{n}} \sum_{k_1=1}^{r+1} \phi_{k_1}^{(r)} .
 \end{eqnarray}
 \end{theorem}

Now, let $\psi_1$ be the following function
\begin{eqnarray}\label{eq:saip1}
\psi_1(\calX,\calY,\p,\q,\m,\beta,s,t) & = & -\frac{\mbox{sign}(s) s \beta}{2\sqrt{n}} \nonumber \\
& & \times \sum_{k=1}^{r+1}\Bigg(\Bigg. \p_{k-1}\q_{k-1}\mE_{G,{\mathcal U}_{r+1}} \langle\|\x^{(i_1)}\|_2\|\x^{(p_1)}\|_2\|\y^{(i_2)}\|_2\|\y^{(p_2)}\|_2\rangle_{\gamma_{k}^{(r)}}\nonumber \\
& & -\p_{k}\q_{k}\mE_{G,{\mathcal U}_{r+1}} \langle\|\x^{(i_1)}\|_2\|\x^{(p_1)}\|_2\|\y^{(i_2)}\|_2\|\y^{(p_2)}\|_2\rangle_{\gamma_{k+1}^{(r)}}\Bigg.\Bigg)
\m_{k} \nonumber \\
& & +\psi(\calX,\calY,\p,\q,\m,\beta,s,t),
 \end{eqnarray}
and consider the following system of equations
\begin{eqnarray}\label{eq:saip4}
\frac{d\psi_1(\calX,\calY,\p,\q,\m,\beta,s,t)}{d\p_{k_1}}
& = &
(1-t)\lp \m_{k_1}-\m_{k_1+1}\rp\frac{\mbox{sign}(s)s\beta}{2\sqrt{n}} \nonumber \\
& & \times \Bigg(\Bigg.   \mE_{G,{\mathcal U}_{r+1}} \langle\|\x^{(i_1)}\|_2\|\x^{(p_1)}\|_2 \lp \q_{k_1} \|\y^{(i_2)}\|_2\|\y^{(p_2)}\|_2 -(\y^{(p_2)})^T\y^{(i_2)}\rp\rangle_{\gamma_{k_1+1}^{(r)}} \Bigg.\Bigg)\nonumber \\
& = & 0, \nonumber \\
\frac{d\psi_1(\calX,\calY,\p,\q,\m,\beta,s,t)}{d\q_{k_1}}
& = &
(1-t)\lp \m_{k_1}-\m_{k_1+1}\rp\frac{\mbox{sign}(s)s\beta}{2\sqrt{n}} \nonumber \\
& & \times \Bigg(\Bigg.   \mE_{G,{\mathcal U}_{r+1}} \langle\|\y^{(i_2)}\|_2\|\y^{(p_2)}\|_2
\lp \p_{k_1} \|\x^{(i_1)}\|_2\|\x^{(p_1)}\|_2 -(\x^{(p_1)})^T\x^{(i_1)}\rp\rangle_{\gamma_{k_1+1}^{(r)}} \Bigg.\Bigg) \nonumber \\
 & = & 0, \nonumber \\
\frac{d\psi_1(\calX,\calY,\p,\q,\m,\beta,s,t)}{d\m_{k_1}}
 & = & 0,
  \end{eqnarray}
where $k_1\in\{1,2,\dots,r\}$ and $\p_0(t)=\q_0(t)=\m_0(t)=1$. The following theorem is then of critical importance.

\begin{theorem}(\cite{Stojnicsflgscompyx23})
\label{thm:thm5}
Assume completely stationirized  fully lifted random duality theory frame (\textbf{\emph{complete sfl RDT frame}}) of \cite{Stojnicsflgscompyx23} with $n\rightarrow\infty$ and $\bar{\p}(t)$, $\bar{\q}(t)$, and $\bar{\m}(t)$ as solutions of system (\ref{eq:saip4}). For $\bar{\p}_0(t)=\bar{\q}_0(t)=1$,
$\bar{\p}_{r+1}(t)=\bar{\q}_{r+1}(t)=\bar{\m}_{r+1}(t)=0$, and $\m_1(t)\rightarrow\m_0(t)=1$ one then has that
$\frac{d\psi_1(\calX,\calY,\bar{\p}(t),\bar{\q}(t),\bar{\m}(t),\beta,s,t)}{dt}   =   0$ and
\begin{eqnarray}\label{eq:thm5eq0}
 \lim_{n\rightarrow\infty}\psi_1(\calX,\calY,\bar{\p}(t),\bar{\q}(t),\bar{\m}(t),\beta,s,t)
& = &\lim_{n\rightarrow\infty}\psi_1(\calX,\calY,\bar{\p}(0),\bar{\q}(0),\bar{\m}(0),\beta,s,0) \nonumber \\
 & = & \lim_{n\rightarrow\infty}\psi_1(\calX,\calY,\bar{\p}(1),\bar{\q}(1),\bar{\m}(1),\beta,s,1).
 \end{eqnarray}
\end{theorem}
\begin{proof}
  Proven in \cite{Stojnicsflgscompyx23}, through a combination of Theorems \ref{thm:thm3} and \ref{thm:thm4}.
\end{proof}
Moreover, we have the following practically relevant corollary as well

\begin{corollary}(\cite{Stojnicsflgscompyx23})
\label{thm:thm6}
Assume the setup of Theorem \ref{thm:thm5}. Let ${\mathcal X}$ and ${\mathcal Y}$ be such that $\|\x\|_2=\|\y\|_2=1$ (or, alternatively, such that their concentrating values in the sfl RDT frame sense of \cite{Stojnicsflgscompyx23} are $1$). Then
\begin{eqnarray}\label{eq:thm6eq0}
 \lim_{n\rightarrow\infty}\psi_1(\calX,\calY,\bar{\p}(t),\bar{\q}(t),\bar{\m}(t),\beta,s,t)
& = &\lim_{n\rightarrow\infty}\psi_1(\calX,\calY,\bar{\p}(0),\bar{\q}(0),\bar{\m}(0),\beta,s,0) \nonumber \\
 & = & \lim_{n\rightarrow\infty}\psi_1(\calX,\calY,\bar{\p}(1),\bar{\q}(1),\bar{\m}(1),\beta,s,1),
 \end{eqnarray}
and
\begin{eqnarray}\label{eq:thm6eq1}
\lim_{n\rightarrow\infty}\psi_1(\calX,\calY,\bar{\p}(0),\bar{\q}(0),\bar{\m}(0),\beta,s,0) & = & -\lim_{n\rightarrow\infty}\frac{\mbox{sign}(s) s \beta}{2\sqrt{n}} \sum_{k=1}^{r+1}\Bigg(\Bigg. \bar{\p}_{k-1}(0)\bar{\q}_{k-1}(0)  -\bar{\p}_{k}(0)\bar{\q}_{k}(0)   \Bigg.\Bigg)
\bar{\m}_k(0) \nonumber \\
& &  +\lim_{n\rightarrow\infty}\psi(\calX,\calY,\bar{\p}(0),\bar{\q}(0),\bar{\m}(0),\beta,s,0) \nonumber \\
\lim_{n\rightarrow\infty}\psi_1(\calX,\calY,\bar{\p}(1),\bar{\q}(1),\bar{\m}(1),\beta,s,1) & = &
-\lim_{n\rightarrow\infty}\frac{\mbox{sign}(s) s \beta}{2\sqrt{n}} \sum_{k=1}^{r+1}\Bigg(\Bigg. \bar{\p}_{k-1}(1)\bar{\q}_{k-1}(1)  -\bar{\p}_{k}(1)\bar{\q}_{k}(1)   \Bigg.\Bigg)
\bar{\m}_k(1)  \nonumber \\
& &  +\lim_{n\rightarrow\infty}\psi(\calX,\calY,\bar{\p}(1),\bar{\q}(1),\bar{\m}(1),\beta,s,1).
 \end{eqnarray}
 Moreover, let
 \begin{equation}\label{eq:thm6eq2}
\psi_S(\calX,\calY,\p,\q,\m,\beta,s,t)  =  \mE_{G,{\mathcal U}_{r+1}} \frac{1}{\beta|s|\sqrt{n}\m_r} \log
\lp \mE_{{\mathcal U}_{r}} \lp \dots \lp \mE_{{\mathcal U}_2}\lp\lp\mE_{{\mathcal U}_1} \lp Z_S^{\m_1}\rp\rp^{\frac{\m_2}{\m_1}}\rp\rp^{\frac{\m_3}{\m_2}} \dots \rp^{\frac{\m_{r}}{\m_{r-1}}}\rp,
\end{equation}
where,  analogously to (\ref{eq:thm3eq1}) and (\ref{eq:thm3eq2}),
\begin{equation}\label{eq:thm6eq3}
Z_S  \triangleq  \sum_{i_1=1}^{l}\lp\sum_{i_2=1}^{l}e^{\beta D_{0,S}^{(i_1,i_2)}} \rp^{s} \nonumber \\
\end{equation}
and
\begin{equation}\label{eq:thm6eq4}
 D_{0,S}^{(i_1,i_2)}  \triangleq  \sqrt{t}(\y^{(i_2)})^T
 G\x^{(i_1)}+\sqrt{1-t}\|\x^{(i_2)}\|_2 (\y^{(i_2)})^T\lp\sum_{k=1}^{r+1}b_k\u^{(2,k)}\rp  +\sqrt{1-t}\|\y^{(i_2)}\|_2\lp\sum_{k=1}^{r+1}c_k\h^{(k)}\rp^T\x^{(i)}.
 \end{equation}
Then
\begin{eqnarray}\label{eq:thm6eq5}
 \lim_{n\rightarrow\infty} \psi_S(\calX,\calY,\bar{\p}(1),\bar{\q}(1),\bar{\m}(1),\beta,s,1) & = &
 -\lim_{n\rightarrow\infty}\frac{\mbox{sign}(s) s \beta}{2\sqrt{n}} \sum_{k=1}^{r+1}\Bigg(\Bigg. \bar{\p}_{k-1}(0)\bar{\q}_{k-1}(0)  -\bar{\p}_{k}(0)\bar{\q}_{k}(0)   \Bigg.\Bigg)
\bar{\m}_k(0)
 \nonumber \\
& &  +\lim_{n\rightarrow\infty}\psi_S(\calX,\calY,\bar{\p}(0),\bar{\q}(0),\bar{\m}(0),\beta,s,0). \nonumber \\
 \end{eqnarray}
\end{corollary}

\subsection{Matching Parisi ansatz rsb predictions}
\label{sec:rrsbmatch}

To make the matching with the replica predictions possible, we assume in this section that $\cX$ and $\cY$ are comprised of unit norm elements. Carefully looking at  $\psi_S(\calX,\calY,\bar{\p}(t),\bar{\q}(t),\bar{\m}(t),\beta,s,1)$ one observes that it exactly matches $f_{sq}(\beta)$, the averaged free energy of bilinearly based Hopfield modes in the thermodynamic limit. In other words one has
\begin{eqnarray}\label{eq:match1}
 \lim_{n\rightarrow\infty}\psi_S(\calX,\calY,\bar{\p}(1),\bar{\q}(1),\bar{\m}(1),\beta,s,1)  = f_{sq}(\beta).
 \end{eqnarray}
 To be able to match the right hand side of (\ref{eq:thm6eq5}) to the replica predictions, we would need concrete values of rsb parametrization vectors $\p_{rsb}=[p_r,p_{r-1},\dots,p_0]$, $\bar{\q}_{rsb}=[\bq_r,\bq_{r-1},\dots,\bq_0]$, and $\m_{rsb}=[m_r,m_{r-1},\dots,m_1]$. Even within the replica theory context, it is an extraordinary challenge to properly choose these sequences or their related alternatives $\alpha$ and $\bar{\alpha}$. We make the following choice
 \begin{eqnarray}\label{eq:match2}
   \frac{d\bar{f}_{sq}^{(r)}(\beta)}{d\p_{rsb}} =0, \qquad   \frac{d\bar{f}_{sq}^{(r)}(\beta)}{d\bar{\q}_{rsb}} = 0, \qquad
   \frac{d\bar{f}_{sq}^{(r)}(\beta)}{d\m_{rsb}} =0.
\end{eqnarray}
First we set $p_{r+1}=\bq_{r+1}=m_{r+1}=1$ and connect the fully broken rsb and complete sfl RDT indexation
\begin{eqnarray}\label{eq:match1a1}
p_k\leftrightarrow \bar{\p}_{r+1-k}, \quad
\bq_k\leftrightarrow \bar{\q}_{r+1-k},\quad  m_k\leftrightarrow \bar{\m}_{r+2-k}.
 \end{eqnarray}
Carefully comparing the key ingredients of (\ref{eq:2rsbalphabar12sqrt})-(\ref{eq:rrsbq}) and $\psi_S(\calX,\calY,\bar{\p}(0),\bar{\q}(0),\bar{\m}(0),\beta,s,0)$ (with trivially adjusted indexation), and keeping in mind how the $\bar{\p}(t),\bar{\q}(t),$ and $\bar{\m}(t)$ functions are chosen within the complete sfl RDT frame, one  observes that $\lim_{n\rightarrow\infty}\psi_S(\calX,\calY,\bar{\p}(0),\bar{\q}(0),\bar{\m}(0),\beta,s,0)$ is equal to the sum of the last two terms in (\ref{eq:rrsbfinalfreeenergycompsqrt}). If one relies on the more compact form (\ref{eq:rrsbfinalfreeenergycomp1sqrt}), that basically means that
 \begin{eqnarray}\label{eq:match3}
 \lim_{n\rightarrow\infty}\psi_S(\calX,\calY,\bar{\p}(0),\bar{\q}(0),\bar{\m}(0),\beta,s,0)  =  \lim_{n\rightarrow \infty}\frac{E_{\z_{P}^{(r+1)}}\log \lp \zeta_{r,P}\rp}{\frac{\beta}{\sqrt{n}}  m_1n}
  +\lim_{n\rightarrow \infty}\frac{E_{\z_{Q}^{(r+1)}}\log \lp \zeta_{r,Q}\rp}{\frac{\beta}{\sqrt{n}}  m_1n}.
 \end{eqnarray}
 Choosing $s=1$ one then also easily has
\begin{eqnarray}\label{eq:match4}
  \frac{\beta^2}{2\frac{\beta}{\sqrt{n}}n}
\lp (1-\bq_rp_r)+ \sum_{k=1}^{r}m_k(\bq_kp_k-\bq_{k-1}p_{k-1}) \rp
& = &   \frac{\beta^2}{2\frac{\beta}{\sqrt{n}}n}
\lp  \sum_{k=1}^{r+1}m_k(\bq_kp_k-\bq_{k-1}p_{k-1}) \rp \nonumber \\
& = &   \frac{\beta}{2\sqrt{n}} \sum_{k=1}^{r+1}\Bigg(\Bigg. \bar{\p}_{k-1}(0)\bar{\q}_{k-1}(0)  -\bar{\p}_{k}(0)\bar{\q}_{k}(0)   \Bigg.\Bigg)
\bar{\m}_k(0). \nonumber \\
 \end{eqnarray}
A combination of (\ref{eq:match3}) and (\ref{eq:match4}) gives
 \begin{eqnarray}\label{eq:match5}
\bar{f}_{sq}^{(r)}(\beta) & =&
\lim_{n\rightarrow\infty}\frac{\beta}{2\sqrt{n}} \sum_{k=1}^{r+1}\Bigg(\Bigg. \bar{\p}_{k-1}(0)\bar{\q}_{k-1}(0)  -\bar{\p}_{k}(0)\bar{\q}_{k}(0)   \Bigg.\Bigg)
 \bar{\m}_k(0) \nonumber \\
 & &  + \lim_{n\rightarrow\infty}\psi_S(\calX,\calY,\bar{\p}(0),\bar{\q}(0),\bar{\m}(0),\beta,s,0). \nonumber \\
 \end{eqnarray}
From (\ref{eq:thm3eq1}) and (\ref{eq:thm3eq2}), (\ref{eq:thm6eq2}) and (\ref{eq:thm6eq4}), one further finds
 \begin{eqnarray}\label{eq:match6}
  \psi_S(\calX,\calY,\bar{\p}(0),\bar{\q}(0),\bar{\m}(0),\beta,1,0)
  =   \psi(\calX,\calY,\bar{\p}(0),\bar{\q}(0),\bar{\m}(0),\beta,1,0). \nonumber \\
 \end{eqnarray}
Combining (\ref{eq:match5}), (\ref{eq:match6}), and the first part of (\ref{eq:thm6eq1}),  we have
\begin{eqnarray}\label{eq:match7}
 \bar{f}^{(r)}_{sq}(\beta)  & = &
 -\lim_{n\rightarrow\infty}
 \frac{ \beta}{2\sqrt{n}} \sum_{k=1}^{r+1}\Bigg(\Bigg. \bar{\p}_{k-1}(0)\bar{\q}_{k-1}(0)  -\bar{\p}_{k}(0)\bar{\q}_{k}(0)   \Bigg.\Bigg)
\bar{\m}_k(0) \nonumber \\
&&
+ \lim_{n\rightarrow\infty}  \psi_S(\calX,\calY,\bar{\p}(0),\bar{\q}(0),\bar{\m}(0),\beta,1,0) \nonumber \\
  & = &
 -\lim_{n\rightarrow\infty}\frac{\beta}{2\sqrt{n}} \sum_{k=1}^{r+1}\Bigg(\Bigg. \bar{\p}_{k-1}(0)\bar{\q}_{k-1}(0)  -\bar{\p}_{k}(0)\bar{\q}_{k}(0)   \Bigg.\Bigg)
\bar{\m}_k(0) \nonumber \\
& &
+\lim_{n\rightarrow\infty}  \psi(\calX,\calY,\bar{\p}(0),\bar{\q}(0),\bar{\m}(0),\beta,1,0) \nonumber \\
   & = &
\lim_{n\rightarrow\infty}   \psi_1(\calX,\calY,\bar{\p}(0),\bar{\q}(0),\bar{\m}(0),\beta,1,0). \nonumber \\
 \end{eqnarray}
A combination of (\ref{eq:saip4}), (\ref{eq:match2}), and (\ref{eq:match7}) ensures that the starting premise (\ref{eq:match1a1}) is indeed correct as well as everything else that followed after that. Finally, a combination of (\ref{eq:thm6eq5}), (\ref{eq:match1}), and (\ref{eq:match5}), shows that the \emph{complete sfl RDT} machinery of \cite{Stojnicnflgscompyx23,Stojnicsflgscompyx23} matches Parisi ansatz based full rsb and also implies $f_{sq}(\beta)= \bar{f}^{(r)}_{sq}(\beta)$. The following theorem summarizes all the above considerations.

\begin{theorem}\label{thm:thm7} Assume the setup of Corollary \ref{thm:thm6} (and therefore, implicitly, of Theorem \ref{thm:thm5}). Consider bilinear Hamiltonian based models and their thermodynamic limit averaged free energy, $f_{sq}(\beta),$  from (\ref{eq:logpartfunsqrt}). Then
\begin{eqnarray}\label{eq:thm7eq1}
f_{sq}(\beta)= \bar{f}^{(r)}_{sq}(\beta)   =
\lim_{n\rightarrow\infty}   \psi_1(\calX,\calY,\bar{\p}(0),\bar{\q}(0),\bar{\m}(0),\beta,1,0),
 \end{eqnarray}
where  $\bar{f}^{(r)}_{sq}(\beta)$ is as in (\ref{eq:rrsbfinalfreeenergycompsqrt}) (or (\ref{eq:rrsbfinalfreeenergycomp1sqrt})) with the replica symmetry breaking parameters as in (\ref{eq:match2}) and $\psi_1(\calX,\calY,\bar{\p}(0),\bar{\q}(0),\bar{\m}(0),\beta,1,0)$ is as in the first part of (\ref{eq:thm6eq1}). In other words, the \emph{complete  sfl RDT} results of \cite{Stojnicnflgscompyx23,Stojnicsflgscompyx23} both \underline{exactly match} the predictions of the \emph{full replica symmetry breaking} under Parisi ansatz with the choice of replica parameters from (\ref{eq:match2}) and \underline{imply} $f_{sq}(\beta)= \bar{f}^{(r)}_{sq}(\beta)$.
\end{theorem}
\begin{proof}
  Follows from the above discussion.
\end{proof}

\section{Generality of the model}
\label{sec:examples}

The generality of the considered bilinear models is discussed to great extent in \cite{Stojnicnflgscompyx23,Stojnicsflgscompyx23}. We refer to \cite{Stojnicnflgscompyx23,Stojnicsflgscompyx23} for a more through overview, and here just briefly mention that the considered models encompass at once a very large number of well known problems that have been the subject of intensive studies over the last several decades in a host of different scientific fields. Some of the well known examples are the so-called Hopfield models (see, e.g., \cite{Hop82,PasFig78,Hebb49,PasShchTir94,ShchTir93,BarGenGueTan10,BarGenGueTan12,Tal98,StojnicMoreSophHopBnds10}), asymmetric Little models
(see, e.g.,  \cite{BruParRit92,Little74,BarGenGue11bip,CabMarPaoPar88,AmiGutSom85,StojnicAsymmLittBnds11}), various neural networks models
including  the \emph{spherical} perceptrons   (see, e.g., \cite{FPSUZ17,FraHwaUrb19,FraPar16,FraSclUrb19,FraSclUrb20,AlaSel20,StojnicGardGen13,StojnicGardSphErr13,StojnicGardSphNeg13,GarDer88,Gar88,Schlafli,Cover65,Winder,Winder61,Wendel62,Cameron60,Joseph60,BalVen87,Ven86,SchTir02,SchTir03}),  as well as the \emph{binary} perceptrons (see, e.g., \cite{StojnicGardGen13,GarDer88,Gar88,StojnicDiscPercp13,KraMez89,GutSte90,KimRoc98}) and many others.

The range of applications of the presented concepts seems rather unlimited. Studying such applications is usually problem specific and we will present many interesting results that can be obtained in these directions in separate papers.

\section{Conclusion}
\label{sec:lev2x3lev2liftconc}

A connection between the replica symmetry breaking (rsb) predictions and random processes comparisons in the bilinear Hamiltonian based models is studied. In  \cite{Stojnicnflgscompyx23} a very powerful \emph{fully lifted} statistical interpolating/comparison mechanism is introduced. It substantially expanded on an earlier partial lifted variant from \cite{Stojnicgscompyx16}. Here we show that the results of  \cite{Stojnicnflgscompyx23} and
its a stationarized realization introduced in \cite{Stojnicsflgscompyx23} are indeed extremely powerful. In particular, in a very generic context and for a rather wide class of bilinear Hamiltonian models, we study their average free energy and show that the machinery of \cite{Stojnicsflgscompyx23,Stojnicnflgscompyx23}  is strong enough  to produce results that \emph{exactly} match the Parisi ansatz based \emph{full rsb} predictions.

While we chose the free energy as the key object of interest, many other quantities associated with the considered models can be studied as well. As this is the introductory paper, where we introduce and present the core of the methodology, we defer a detailed discussion related to particular extensions to companion papers. We just mention briefly that besides being very powerful when handling problems considered here, the obtained results are also very generic as they can be applied to a plethora of well known scenarios appearing in various random structures and optimization problems. Problem specific adjustments that such applications often require, will be presented in separate papers.

\begin{singlespace}
\bibliographystyle{plain}
\bibliography{nflgscompyxRefs}
\end{singlespace}

\end{document}